\documentclass[a4paper, 10pt]{article}

\usepackage{bbm,xspace}

\usepackage{amsmath}
\usepackage{amssymb}
\usepackage{amsthm}
\usepackage{color}

\usepackage[margin=3.2cm]{geometry}

\theoremstyle{plain}
\newtheorem{theorem}{Theorem}[section]

\theoremstyle{definition}

\newtheorem{remark}{Remark}[section]

\newcommand{\as}{\\[.6em]}

\newcommand{\AS}{\\[1.2em]}
\newcommand{\dis}{\displaystyle}
\newcommand{\bela}[1]{\begin{equation}\label{#1}}
\newcommand{\ela}{\end{equation}}
\newcommand{\bear}[1]{\begin{array}{#1}}
\newcommand{\ear}{\end{array}}
\newcommand{\del}{\partial}

\newcommand{\pf}{\operatorname{pf}}

\newcommand{\rank}{\operatorname{rank}}

\newcommand{\di}{\operatorname{div}}
\newcommand{\fb}{\mbox{\boldmath$f$}}
\newcommand{\high}[1]{#1}
\newcommand{\highb}[1]{#1}
\renewcommand{\Xi}{\mathsf{H}}
\newcommand{\vol}{\mathrm{vol}_4}

\numberwithin{equation}{section}

\title{\high{Self-dual Einstein spaces and the general heavenly equation. Eigenfunctions as coordinates}}

\author{B.G.\ Konopelchenko$^{1}$, W.K.\ Schief$^{2}$ and A.\ Szereszewski$^{3}$
\bigskip\\  
$^{1}$Department of Mathematics and Physics,\\ University of Salento and 
INFN, Sezione di Lecce,\\ Lecce, 73100, Italy
\bigskip\\
$^{2}$School of Mathematics and Statistics,\\ The University of New South Wales,\\ Sydney, NSW 2052, Australia
\bigskip\\
$^{3}$Institute of Theoretical Physics,\\ Faculty of Physics, University of Warsaw, \high{Pasteura~5, 02-093 Warsaw}, Poland}

\begin{document}

\maketitle

\begin{abstract}
Eigenfunctions are shown to constitute privileged coordinates of self-dual Einstein spaces with the underlying governing equation being revealed as the general heavenly equation. The formalism developed here may be used to link algorithmically a variety of known heavenly equations. In particular, the classical connection between Pleba\'nski's first and second heavenly equations is retrieved and interpreted in terms of eigenfunctions. In addition, connections with travelling wave reductions of the recently introduced TED equation which constitutes a 4+4-dimensional integrable generalisation of the general heavenly equation are found. These are obtained by means of (partial) Legendre transformations. As a particular application, we prove that a large class of self-dual Einstein spaces governed by a compatible system of dispersionless Hirota equations is genuinely four-dimensional in that the (generic) metrics do not admit any (proper or non-proper) conformal Killing vectors. This generalises the known link between a particular class of self-dual Einstein spaces and the dispersionless Hirota equation encoding three-dimensional Einstein-Weyl geometries. 
\end{abstract}

\section{Introduction}

A variety of important solution generation techniques applied to Einstein's field equations have their origin in (or may be interpreted in terms of) integrable systems theory (see \cite{SKMHH2003, RogersSchief2002} and references therein). Eigenfunctions play a key role in integrable systems theory since they encode integrable systems via the compatibility of the linear systems (Lax pairs) which they satisfy. Eigenfunctions, in turn, obey partial differential equations (eigenfunction equations) which are intimately related to the important Miura-type transformations. The latter provide a link between integrable hierarchies and their associated modified versions such as the Korteweg-de Vries (KdV) and modified KdV (mKdV) hierarchies. Eigenfunctions are the key ingredient in Darboux-type transformations which lie at the heart of solution generation techniques (B\"acklund transformations) and their associated superposition principles (permutability theorems) for solutions of integrable systems. Eigenfunctions also encode conserved quantities associated with conservation laws hidden in integrable systems. Details on the above subjects and corresponding references may be found in, e.g., \cite{RogersSchief2002,AblowitzSegur1981,AblowitzClarkson1991,Fordy1990,Konopelchenko1990}. 

Hierarchies of 1+1-dimensional modified integrable equations are known to admit integrable counterparts which are related by reciprocal transformations. A geometrically important example is the integrable nonlinear Schr\"odinger (NLS) equation, the modified version of which is given by the Heisenberg spin equation which, in turn, gives rise to the reciprocally related loop soliton equation (see, e.g., \cite{RogersSchief2002} and references therein). In the context of integrable systems, the analogue of reciprocal transformations for 2+1-dimensional hierarchies has been shown to involve eigenfunctions which, importantly, constitute some of the new independent variables in the ``reciprocally'' related hierarchies. A prime example of the link between the associated three types of hierarchies of 2+1-dimensional integrable equations is provided by the connection between the Kadomtsev-Petviashvili (KP) hierarchy, its modified (mKP) hierarchy and the reciprocally related 2+1-dimensional Dym hierarchy. In this connection, the reader may wish to consult \cite{OevelRogers1993} for details and references.

In this paper, we present a self-contained first application to general relativity of a general scheme to be discussed elsewhere which, in any dimension, employs eigenfunctions of dispersionless integrable systems as independent variables in the corresponding privileged equivalent systems. Even though dispersionless integrable systems have been studied extensively (see, e.g., \cite{BogdanovKonopelchenko2013, ManakovSantini2006} and references therein) since the pioneering work of Zakharov and Shabat \cite{ZakharovShabat1979}, the interpretation of eigenfunctions as independent variables has not been at the forefront of these investigations. Here, we demonstrate that any four eigenfunctions of the self-dual Einstein equations \cite{Boyer1983} corresponding to four distinct spectral parameters give rise to a unique representation of self-dual Einstein spaces in terms of a potential which depends on the eigenfunctions playing the role of coordinates. Remarkably, the associated underlying dispersionless equation turns out to be the general heavenly equation introduced in \cite{Schief1996, Schief1999} as proven in Section 4. It is noted in passing that particular coefficients of the Laurent expansions of eigenfunctions associated with the self-dual Einstein equations have been identified in \cite{Schief1996,Takasaki1989} as the independent variables in Pleba\'nski's important first and second heavenly equations~\cite{Plebanski1975} governing self-dual Einstein spaces. An invariant definition of the above-mentioned key potential is presented in Section 5.

It turns out that the assumption of one or two pairs of coinciding spectral parameters is also admissible. This is discussed in Sections 7 and 8 respectively. In this manner, on the one hand, we retrieve the connection between the first heavenly equation and the Husain-Park equation \cite{JakimowiczTafel2006} and, on the other hand, we demonstrate that the classical connection between Pleba\'nski's two heavenly equations \cite{Plebanski1975} may be viewed as a particular application of the algorithm developed in this paper. In Section 12, the classification in terms of (0, 1 or 2) pair(s) of coinciding spectral parameters is then shown to go hand in hand with travelling wave reductions of the recently introduced TED equation which constitutes a 4+4-dimensional integrable extension of the general heavenly equation \cite{KonopelchenkoSchief2019}. The connection between these two {\em a priori} unrelated subjects is provided by (partial) Legendre transformations. In particular, in Section 6, the latter is shown to leave invariant the general heavenly equation. \high{The investigation of Legendre transformations is motivated by the important observation that the gradient of the potential satisfying the general heavenly equation is, in fact, composed of another four eigenfunctions associated with the same spectral parameters.}

As a further application of the algorithm presented here, we show that the general heavenly equation may be specialised to a system of four compatible dispersionless Hirota equations \cite{Krynski2018} by matching the eigenfunction and scaling symmetries of the general heavenly equation \cite{KonopelchenkoSchief2019, Sergyeyev2017}. It turns out that the self-dual Einstein metrics obtained in this manner generalise those associated with the three-dimensional Einstein-Weyl geometries known to be governed by the dispersionless Hirota equation \cite{DunajskiKrynski2014} as discussed in Section 9. In general, in Section 10, the metrics generated by generic solutions of the dispersionless Hirota system are proven to be genuinely four-dimensional in the sense that no conformal Killing vectors (including homothetic Killing vectors and Killing vectors) exist. The key to the proof of this property is the derivation of the action of the above-mentioned eigenfunction symmetry constraint on the first heavenly equation, leading to a decomposition into three compatible differential equations. This is achieved in Section~11 by exploiting the connection between the general heavenly equation and Pleba\'nski's first heavenly equation derived in Section 8.

\section{The equations governing self-dual Einstein spaces}

Self-dual Einstein spaces constitute four-dimensional manifolds \high{equipped with a metric} which are characterised by a self-dual Riemann tensor. Self-duality implies that the Ricci tensor $R_{ik}$ vanishes and, hence, Einstein's vacuum equations $R_{ik}=0$ are indeed satisfied \cite{Boyer1983}. Remarkably, the equations governing self-dual Einstein spaces have been shown to be equivalent to the self-dual Yang-Mills equations with four translational symmetries and the gauge group of volume preserving diffeomorphisms \cite{MasonNewman1989}. The latter are encoded in the commutativity of two four-dimensional vector fields $X$ and $Y$ which are linear in a (spectral) parameter $\lambda$ and divergence free with respect to a volume form \cite{AblowitzClarkson1991}. Specifically, consider two commuting vector fields
\bela{N1}
  X(\lambda) = A_1 + \lambda A_2,\quad Y(\lambda) = A_3 + \lambda A_4
\ela
so that the commutativity condition $[X,Y] = 0$ is equivalent to
\bela{N2}
 [A_1,A_3] = 0,\quad [A_2,A_4] = 0,\quad [A_1,A_4] + [A_2,A_3] = 0.
\ela
The latter constitute partial differential equations for the coefficients $A_{\alpha}^i$ of the vector fields
\bela{N3}
  A_{\alpha} = A_{\alpha}^i\del_{x^i},
\ela
where we have adopted Einstein's summation convention over repeated indices. In addition, we assume that the vector fields $X$ and $Y$ are divergence free with respect to a volume form
\bela{N4}
  \highb{\vol = fdx^1\wedge dx^2\wedge dx^3\wedge dx^4}
\ela
encoded in a function $f(x^i)$ so that
\bela{N5}
  f\di A_{\alpha} = \del_{x^i}(fA_{\alpha}^i) = 0.
\ela
The commutativity of $X$ and $Y$ implies the compatibility of the Lax pair \cite{AblowitzClarkson1991}
\bela{N6}
  X\Phi = 0,\quad Y\Phi = 0.
\ela

Different choices of the coordinates $x^i$ lead to different but equivalent forms of the self-dual Einstein equations represented by the system \eqref{N2}, \eqref{N5}. For instance, Pleba\'nski's celebrated first heavenly equation \cite{Plebanski1975} 
\bela{N6a}
 \Omega_{x^1x^3}\Omega_{x^2x^4} - \Omega_{x^2x^3}\Omega_{x^1x^4} = 1,
\ela
which was, in fact, already recorded in 1936 in the context of ``wave geometry'' \cite{SibataMorinaga1936}, corresponds to the choice
(see, e.g., \cite{DoubrovFerapontov2010})
\bela{N6b}
 \begin{aligned}
 A_1 & = - \del_{x^3},\quad & A_2 & = \Omega_{x^1x^3}\del_{x^2} - \Omega_{x^2x^3}\del_{x^1}\\
 A_3 & = - \del_{x^4},\quad & A_4 & = \Omega_{x^1x^4}\del_{x^2} - \Omega_{x^2x^4}\del_{x^1}.
 \end{aligned}
\ela
The latter vector fields are indeed seen to be divergence free with respect to the function $f=1$ and one may verify that the commutator relations \eqref{N2} modulo the first heavenly equation \eqref{N6a} are satisfied. The corresponding metric reads \cite{Plebanski1975}
\bela{N6c}
  g = 2\Omega_{x^1x^3}dx^1dx^3 + 2\Omega_{x^2x^3}dx^2dx^3 + 2\Omega_{x^1x^4}dx^1dx^4 + 2\Omega_{x^2x^4}dx^2dx^4
\ela
for which the Ricci tensor $R_{ik}$ may be shown to vanish.

\section{Eigenfunctions as privileged coordinates}

Within the general setting \eqref{N1}-\eqref{N6} of self-dual Einstein spaces, we now consider four \high{(functionally independent)} eigenfunctions $\phi^i$ corresponding to four {\em distinct} parameters $\lambda_i$, that is,
\bela{N7}
  X(\lambda_i)\phi^i = 0,\quad Y(\lambda_i)\phi^i = 0,\quad i=1,2,3,4
\ela
(no summation over the index $i$ as explained below). The important cases of one pair and two pairs of coinciding parameters are dealt with in Sections 7 and 8 respectively. In terms of the new coordinates 
\bela{N8}
 y^i = \phi^i,
\ela
the vector fields $X$ and $Y$ are represented by
\bela{N9}
 \begin{aligned}
  X &= (X\phi^i)\del_{y^i} = (\lambda - \lambda_i)(A_2\phi^i)\del_{y^i}\\
  Y &= (Y\phi^i)\del_{y^i}\,= (\lambda - \lambda_i)(A_4\phi^i)\del_{y^i}
 \end{aligned}
\ela
and the volume form becomes
\bela{N10}
 \highb{\vol = \tilde{f}dy^1\wedge dy^2\wedge dy^3\wedge dy^4,\quad \tilde{f}=\frac{f}{J},\quad J = \det\left(\frac{\del y^i}{\del x^k}\right).}
\ela
Throughout this paper, indices on the parameters $\lambda_i$ are not taken into account when Einstein's summation convention is applied. Thus, expressions involving two repeated indices such as $\lambda_i a^i$ in which one index is attached \high{to} $\lambda$ do not denote a sum but expressions such as
\bela{N10a}
  \lambda_ia_ib^i = \sum_{i=1}^4
\lambda_ia_ib^i
\ela
encode summation over the index $i$. In particular, Einstein's summation convention applies to the vector field representation \eqref{N9}. The latter naturally leads to the introduction of the four vector fields
\bela{N11}
  \tilde{X}^{(i)} = (A_4\phi^i) X - (A_2\phi^i) Y.
\ela
Remarkably, even though the vector fields $\tilde{X}^{(i)}$ constitute linear combinations of the vector fields $X$ and $Y$ with {\em non-constant} coefficients, the following theorem obtains.

\begin{theorem}
 The vector fields $\tilde{X}^{(i)}$ are divergence free.
\end{theorem}

\begin{proof}
The proof of this theorem is based on the general (obvious) fact that if two commuting vector fields $X$ and $Y$ are divergence free then the vector field $(Yg)X - (Xg)Y$ is likewise divergence free for any function $g$. Hence, if we set $g = \phi^i$ then
\bela{N12}
  0 = \di [(Y\phi^i)X - (X\phi^i)Y] = (\lambda-\lambda_i)\di [(A_4\phi^i) X - (A_2\phi^i) Y]
\ela
since $(A_1+\lambda_iA_2)\phi^i=0$ and $(A_3+\lambda_iA_4)\phi^i=0$.
\end{proof}

In terms of components, the vector fields $\tilde{X}^{(i)}$ adopt the form
\bela{N13}
  \tilde{X}^{(i)} = (\lambda-\lambda_k)\tilde{A}^{ik}\del_{y^k},\quad \tilde{A}^{ik} = (A_4\phi^i)(A_2\phi^k )- (A_2\phi^i)(A_4\phi^k).
\ela
It is observed that the coefficients of the skew-symmetric matrix $\tilde{A}$ may be regarded as Pl\"ucker coordinates \cite{HodgePedoe1994} of a line
\bela{N14}
  A_4\left(\bear{c}\phi^1\\ \phi^2\\ \phi^3\\ \phi^4\ear\right)\wedge A_2\left(\bear{c}\phi^1\\ \phi^2\\ \phi^3\\ \phi^4\ear\right)
\ela
in a three-dimensional projective space represented in terms of homogeneous coordinates. Accordingly, the original Lax pair \eqref{N6} may equivalently be formulated as the set of four equations
\bela{N15}
  \tilde{f}\tilde{X}^{(i)}\Phi = A^{ik}D_k\Phi = 0,\quad A^{ik} = \tilde{f}\tilde{A}^{ik},\quad D_k = (\lambda-\lambda_k)\del_{y^k}
\ela
since $\rank A = \rank \tilde{A}=2$. The latter is reflected by the Pl\"ucker relation
\bela{N16}
 \pf(A) = A^{12}A^{34} + A^{23}A^{14} + A^{31}A^{24} = 0,
\ela
where $\pf(A)$ denotes the Pfaffian of $A$. The relevance of this observation is revealed below. In particular, the significance of the scaling of the matrix $\tilde{A}$ is explained.

\section{The general heavenly equation}

We are now in a position to present the key theorem of this paper.

\begin{theorem}\label{key}
Let $\phi^i$, $i=1,2,3,4$ be four eigenfunctions associated with any vector field representation $X(x^i;\lambda)$, $Y(x^i;\lambda)$ of self-dual Einstein spaces and distinct parameters $\lambda_i$. \high{Then, in terms of the independent variables $y^i=\phi^i$, the self-dual Einstein equations are transformed into the general heavenly equation
\bela{N23}
 \begin{split}
  & \phantom{+}\,\,\, (\lambda_1-\lambda_2)(\lambda_3-\lambda_4)\Theta_{y^1y^2}\Theta_{y^3y^4}\\
  & +  (\lambda_2-\lambda_3)(\lambda_1-\lambda_4)\Theta_{y^2y^3}\Theta_{y^1y^4}\\  
  & +  (\lambda_3-\lambda_1)(\lambda_2-\lambda_4)\Theta_{y^3y^1}\Theta_{y^2y^4} = 0
\end{split}
\ela
for some potential $\Theta$ encoded in the divergence-free vector fields $X$ and $Y$.}
\end{theorem}

\begin{remark}
The differential equation \eqref{N23}, which may be formulated as
\bela{N23a}
 c_1 \Theta_{y^1y^2}\Theta_{y^3y^4} +  c_2 \Theta_{y^2y^3}\Theta_{y^1y^4} +  c_3 \Theta_{y^3y^1}\Theta_{y^2y^4} = 0,\quad c_1 + c_2 + c_3 = 0,
\ela
was originally derived as the continuum limit of the permutability theorem for both the classical Tzitz\'eica equation of affine differential geometry and its integrable discretisation \cite{Schief1996,Schief1999,BobenkoSchief1999}. In this context, it was also observed that this equation constitutes yet another avatar of the self-dual Einstein equations. In connection with the classification of integrable symplectic Monge-Amp\'ere equations, it was rediscovered and termed general heavenly equation \cite{DoubrovFerapontov2010}. The above theorem demonstrates that the general heavenly equation is privileged in that the associated coordinates play the role of eigenfunctions regardless of the concrete realisation of the vector fields $A_{\alpha}$.
\end{remark}

In order to verify the above theorem, we begin by noting that the scaling of the matrix $\tilde{A}$, leading to the matrix $A$ as defined by \eqref{N15}$_2$, has been introduced for convenience so that the vanishing divergence of the vector fields $\tilde{X}^{(i)}$ simplifies to the system of equations
\bela{N17}
  \del_{y^k}A^{ik} = 0,\quad \lambda_k\del_{y^k}A^{ik} = 0.
\ela
In terms of the skew-symmetric matrix $\omega$ defined by
\bela{N18}
 A^{lm} = \epsilon^{iklm}\omega_{ik},
\ela
where $\epsilon^{iklm}$ denotes the totally anti-symmetric Levi-Civita symbol,
these adopt the form
\bela{N19}
 \del_{y^{[i}}\omega_{kl]} = 0,\quad  \lambda_{[i}\del_{y^i}\omega_{kl]} = 0
\ela
and may be regarded as the integrability conditions for the existence of a potential $\Theta$ defined by
\bela{N19a}
 \Theta_{y^iy^k} = \frac{\omega_{ik}}{\lambda_k-\lambda_i}.
\ela
Here, the square brackets indicate total anti-symmetrisation. Accordingly, the solution of \eqref{N19} is parametrised in terms of $\Theta$ according to
\bela{N20}
 \omega_{ik} = (\lambda_k - \lambda_i)\Theta_{y^iy^k}.
\ela
Now, the linear system \eqref{N15}$_1$ may be formulated as 
\bela{N21}
  \left(\bear{cccc} 0 & \omega_{34} & \omega_{42} & \omega_{23}\\ \omega_{43} & 0 & \omega_{14} & \omega_{31}\\
                    \omega_{24} & \omega_{41} & 0 & \omega_{12}\\ \omega_{32} & \omega_{13} & \omega_{21} & 0\ear\right)
  \left(\bear{c}D_1\Phi\\ D_2\Phi\\ D_3\Phi\\ D_4\Phi\ear\right)  = 0
\ela
with the associated rank 2 condition $\pf(A)=0$ expressed as
\bela{N22}
  \pf(\omega) = \omega_{12}\omega_{34} + \omega_{23}\omega_{14} + \omega_{31}\omega_{24} = 0.
\ela
The parametrisation \eqref{N20} therefore leads to the general heavenly equation \eqref{N23} with, for instance,
\bela{N24}
  \tilde{X}^{(3)}\Phi = 0,\quad \tilde{X}^{(4)}\Phi = 0
\ela
constituting its standard Lax pair \cite{Schief1996}. Here, it should be emphasised that the scaled vector fields
\bela{N25}
  \hat{X}^{(3)} = \frac{\omega_{24}}{\omega_{12}}D_1 + \frac{\omega_{41}}{\omega_{12}}D_2 + D_4,\quad
  \hat{X}^{(4)} = \frac{\omega_{32}}{\omega_{21}}D_1 + \frac{\omega_{13}}{\omega_{21}}D_2 + D_3
\ela
commute and are divergence free with respect to the volume form 
\bela{N26}
 \highb{\vol = \Theta_{y^1y^2}dy^1\wedge dy^2\wedge dy^3\wedge dy^4}
\ela
as originally observed in \cite{Schief1996}. Accordingly, the general heavenly equation may be regarded as an ``invariant'' form of the self-dual Einstein equations since neither the eigenfunctions $\phi^i$ nor the potential $\Theta$ are affected by coordinate transformations $x^i\rightarrow f^i(x^k)$. In particular, any of the known ``heavenly'' equations governing self-dual Einstein spaces (see, e.g., \cite{DoubrovFerapontov2010, PlebanskiPrzanowski1996} and references therein) may be mapped to the general heavenly equation by means of the algorithm presented in this section.

\section{Invariant definition of the potential \boldmath $\Theta$}

In order to interpret the potential $\Theta$ within the original self-dual Einstein setting associated with the vector fields $A_{\alpha}(x^i)$, we first \high{make the following important observation.}

\begin{remark}
\high{The dependent variable $\Theta$ also encodes eigenfunctions in that} the quantities
\bela{N27}
 \psi_i = \Theta_{y^i}
\ela
form another set of eigenfunctions corresponding to the parameters $\lambda_i$. Indeed, since
\bela{N28}
  \left.D_k\psi_i\right|_{\lambda=\lambda_i} = (\lambda_i-\lambda_k)\Theta_{y^iy^k} = \omega_{ki},\quad i\neq k,
\ela
three equations of the linear system \eqref{N21} for $\Phi=\psi_i$ and $\lambda=\lambda_i$ are identically satisfied, while the remaining equation turns out to be $\pf(\omega)=0$. It is remarked that this implies that the general solution of the linear system at $\lambda=\lambda_i$ is an arbitrary function of the particular eigenfunctions $\phi^i=y^i$ and $\psi_i = \Theta_{y^i}$.
\end{remark}

The following theorem provides the basis of this section.

\begin{theorem}
Let $\phi^i$, $i=1,\ldots,4$ be four eigenfunctions of the self-dual Einstein equations for distinct parameters $\lambda_i$. Then, the two-forms $\Omega_k$, $k=1,\ldots,4$ defined by
\bela{N29}
 \Omega_k = \frac{1}{4}\epsilon_{i\underline{k}lm}\frac{A^{lm}}{\lambda_{\underline{k}}-\lambda_i}d\phi^i\wedge d\phi^{\underline{k}},
\ela
where 
\bela{N30}
  A^{lm} = \frac{f}{J}\tilde{A}^{lm},\quad \tilde{A}^{lm} = (A_4\phi^l)(A_2\phi^m) - (A_2\phi^l) (A_4\phi^m), \quad J = \det\left(\frac{\del \phi^i}{\del x^k}\right),
\ela
have the following properties:
\bela{N31}
  d\Omega_k = 0,\quad \Omega_{\underline{k}}\wedge\Omega_{\underline{k}} = 0,\quad \Omega_{\underline{k}}\wedge d\phi^{\underline{k}} = 0,\quad \sum_{k=1}^4 \Omega_k = 0.
\ela
Here, the underbar indicates that there is no summation over the corresponding index. 
\end{theorem}

\begin{proof}
Since the algebraic properties \eqref{N31}$_{2,3,4}$ are evident, it remains to show that $d\Omega_k = 0$. In terms of the coordinates $y^i$, the latter is given by 
\bela{N32}
 \epsilon_{i\underline{k}lm}\frac{\del_{y^n}A^{lm}}{\lambda_{\underline{k}}-\lambda_i}dy^n\wedge dy^i\wedge dy^{\underline{k}} = 0.
\ela
Hence, for any fixed $k$, the contribution of the terms proportional to $dy^{m_0}\wedge dy^{n_0}\wedge dy^k$ reads
\bela{N33}
 \frac{\del_{y^{m_0}}A^{l_0m_0}}{\lambda_k-\lambda_{n_0}} +  \frac{\del_{y^{n_0}}A^{l_0n_0}}{\lambda_k-\lambda_{m_0}} = 0,
\ela
wherein the indices $k,l_0,m_0,n_0$ are \high{fixed and} distinct. On clearing the denominators of the above, this may be formulated as
\bela{N34}
 (\lambda_k-\lambda_m)\del_{y^m}A^{l_0m} = 0
\ela
since the terms for $m=k$ and $m=l_0$ \high{in the above sum} vanish identically. Finally, the vanishing divergence conditions \eqref{N17} imply that the relations \eqref{N34} are indeed satisfied and, hence, the two-forms $\Omega_k$ are closed.
\end{proof}

The above theorem encodes the existence of a potential $\Theta$ which coincides with that derived in Section 4 in connection with the parametrisation of the skew-symmetric matrix $A$ in terms of the coordinates $y^i$.

\begin{theorem}
There exist functions $\psi_k$, $k=1,\ldots,4$ such that
\bela{N35}
 \Omega_k = d\psi_{\underline{k}}\wedge d\phi^{\underline{k}}.
\ela
These constitute eigenfunctions for $\lambda=\lambda_k$ and give rise to the existence of a potential $\Theta$ via
\bela{N36}
  d\Theta  = \psi_k d\phi^k.
\ela
\end{theorem}

\begin{proof}
By virtue of Darboux's theorem, the properties \eqref{N31}$_{1,2}$ imply that the two-forms $\Omega_k$ may be written as exterior products of pairs of differentials. Property \eqref{N31}$_3$ then shows that $\Omega_k$ is of the form \eqref{N35} for some function $\psi_k$. The compatibility condition $d\psi_k\wedge d\phi^k =0$ associated with \eqref{N36} guaranteeing the existence of the potential $\Theta$ is satisfied since
\bela{N38}
  d\psi_k\wedge d\phi^k = \sum_{k=1}^4 \Omega_k = 0
\ela
by virtue of \eqref{N31}$_4$. Furthermore, since, in this case,
\bela{N39}
 \Theta_{y^k} = \psi_k,
\ela
comparison of \eqref{N29} and \eqref{N35} results in
\bela{N40}
 \Theta_{y^iy^k} =  \frac{1}{4}\epsilon_{iklm}\frac{A^{lm}}{\lambda_k-\lambda_i}
\ela
which \high{coincides with \eqref{N19a} by virtue of \eqref{N18}} so that \high{$\psi_k = \Theta_{y^k}$} indeed constitutes an eigenfunction for $\lambda=\lambda_k$. It is emphasised that, on use of the pair \eqref{N29}, \eqref{N35} regarded as a definition of the functions $\psi_k$, one may also directly show that $X(\lambda_k)\psi_k = Y(\lambda_k)\psi_k=0$.
\end{proof}

\section{A Legendre transformation}

It is evident that Theorem \ref{key} applied to the general heavenly equation
\bela{N41}
 \begin{split}
  & \phantom{+}\,\,\, (\mu_1-\mu_2)(\mu_3-\mu_4)\Xi_{x^1x^2}\Xi_{x^3x^4}\\
  & +  (\mu_2-\mu_3)(\mu_1-\mu_4)\Xi_{x^2x^3}\Xi_{x^1x^4}\\  
  & +  (\mu_3-\mu_1)(\mu_2-\mu_4)\Xi_{x^3x^1}\Xi_{x^2x^4} = 0
\end{split}
\ela
provides an invariance of the general heavenly equation since it maps \eqref{N41} to the general heavenly equation \eqref{N23}. In general, the associated spectral parameters $\lambda_i$ do not have to be the four parameters $\mu_i$ in the general heavenly equation \eqref{N41} but if we do make this special choice then the corresponding eigenfunctions are given by
\bela{N42}
  \phi^i = f^i(x^i,\Xi_{x^i})
\ela
as pointed out in the previous section. In particular, the choice
\bela{N43}
 \phi^i = \Xi_{x^i}
\ela
is admissible. This raises the question as to whether the quantities $x^i$ and $x_i := \Xi_{x^i}$ play symmetric roles in the general heavenly equation \eqref{N41}. In order to demonstrate that this is the case, we observe that the general heavenly equation may be formulated as
\bela{N44}
 \begin{split}
  &  \phantom{+}\,\,\,  (\mu_2-\mu_3)(\mu_1-\mu_4)(\Xi_{x^1x^2}\Xi_{x^3x^4} - \Xi_{x^2x^3}\Xi_{x^1x^4})\\  
  & +  (\mu_3-\mu_1)(\mu_2-\mu_4)(\Xi_{x^1x^2}\Xi_{x^3x^4} - \Xi_{x^3x^1}\Xi_{x^2x^4}) = 0
\end{split}
\ela
and state the following theorem.

\begin{theorem}
The general heavenly equation in the form
\bela{N45}
 \begin{split}
  &  \phantom{+}\,\,\,  (\mu_1-\mu_4)(\mu_2-\mu_3)dx_1\wedge dx_3 \wedge dx^1 \wedge dx^3\\ 
  & + (\mu_1-\mu_3)(\mu_2-\mu_4)dx_1\wedge dx_4 \wedge dx^1 \wedge dx^4 = 0,\qquad d\Xi = x_i dx^i
 \end{split}
\ela
is invariant under the Legendre transformation
\bela{N46}
  \high{\Xi(x^k) \rightarrow \tilde{\Xi}(x_k)} = x_ix^i - \Xi.
\ela
\end{theorem}

\begin{proof}
In terms of the new variables $\tilde{\Xi}$, $x_i$, relation \eqref{N45}$_2$ becomes
\bela{N47}
 d\tilde{\Xi} = x^idx_i
\ela
which proves the invariance since the remaining relation \eqref{N45}$_1$ is symmetric in the upper and lower indices.
\end{proof}

It is remarked that the general heavenly equation is also invariant under a ``partial'' Legendre transformation which interchanges any chosen number of corresponding variables carrying upper and lower indices, that is, the summation over $i$ in \eqref{N46} may be restricted to any subset of \high{$\{1,\ldots,4\}$} \high{with $\tilde{\Xi}$ depending on the associated appropriate variables}. Partial Legendre transformations are further discussed in Section 12.

\section{A pair of coinciding spectral parameters. The Husain-Park equation}

For any fixed spectral parameter $\lambda$, the general eigenfunction is a function of two functionally independent particular eigenfunctions. Hence, instead of demanding that all parameters $\lambda_i$ be distinct, it is also admissible to choose up to two pairs of coinciding spectral parameters. In this section, we consider the case of {\em one} pair of coinciding parameters, that is, $\lambda_1=\lambda_2$ without loss of generality. Thus, the task is to find the analogue of the parametrisation \eqref{N20} which resolves the vanishing divergence conditions \eqref{N19}. To this end, we observe that the defining equations \eqref{N19a} for the potential $\Theta$ are still meaningful and compatible as long as $(1,2)\neq(i,k)\neq (2,1)$. Accordingly, we may adopt the parametrisation of the generic case, that is,
\bela{N48}
 \omega_{ik} = (\lambda_k-\lambda_i)\Theta_{y^iy^k}
\ela
with the coefficients $\omega_{12}$ and $\omega_{21}=-\omega_{12}$ being excluded. The latter are determined by the remaining vanishing divergence conditions which read
\bela{N49}
  \del_{y^3}\omega_{12} = 0,\quad \del_{y^4}\omega_{12} = 0 
\ela
so that the vanishing Pfaffian condition \eqref{N22} becomes
\bela{N50}
 f_{12}\Theta_{y^3y^4} + \Theta_{y^2y^3}\Theta_{y^1y^4} - \Theta_{y^3y^1}\Theta_{y^2y^4} = 0,
\ela
where
\bela{N51}
 f_{12} = \frac{\high{\lambda_4-\lambda_3}}{(\lambda_1-\lambda_3)(\lambda_1-\lambda_4)}\omega_{12}(y^1,y^2).
\ela
Hence, on application of an appropriate member of the class of coordinate transformations $(y^1,y^2)\rightarrow\fb_{12}(y^1,y^2)$, the Husain-Park equation corresponding to $f_{12}=1$ is obtained (see \cite{JakimowiczTafel2006} and references therein). It is noted that this is consistent with the fact that any function of two eigenfunctions corresponding to the same spectral parameter constitutes another eigenfunction so that it is {\em a priori} known that the differential equation \eqref{N50} must be invariant under this class of coordinate transformations. Furthermore, it should be emphasised that rather than considering the coinciding pair $\lambda_1=\lambda_2$ and applying the above algorithm, one may apply a confluence limit of the type
\bela{N52}
  \Theta\rightarrow \Theta + \frac{g_{12}(y^1,y^2)}{\epsilon},\quad \lambda_2 = \lambda_1 + \epsilon,\quad \epsilon\rightarrow0
\ela
directly to the general heavenly equation \eqref{N23} to derive the Husain-Park equation in the ``gauge-invariant'' form \eqref{N50}.

Before we present concrete examples of how the above procedure may be applied, it is worth noting that the vector fields $X(\lambda_1)$ and $Y(\lambda_1)$ are tangent to the coordinate surfaces $(y^1,y^2)=\mbox{const}$ since
\bela{N52a}
 X(\lambda_1)y^i =0,\quad Y(\lambda_1)y^i=0,\quad i=1,2.
\ela
Hence, the commutativity of $X(\lambda_1)$ and $Y(\lambda_1)$ guarantees the existence of a natural coordinate system $(z^1,z^2,y^1,y^2)$ defined by the additional relations
\bela{N52b}
  X(\lambda_1)z^1 = 1,\quad X(\lambda_1)z^2=0,\quad Y(\lambda_1)z^1 = 0,\quad Y(\lambda_1)z^2 = 1
\ela
in terms of which these two vector fields are ``straight'', that is, 
\bela{N52c}
  X(\lambda_1)=\del_{z^1},\quad Y(\lambda_1)=\del_{z^2}.
\ela
It is observed that even though, for any fixed spectral parameter, say, $\lambda_1$, the corresponding eigenfunction is an arbitrary function of two functionally independent particular eigenfunctions, the kernel of each of the vector fields $X(\lambda_1)$ and $Y(\lambda_1)$ is encoded in three functionally independent solutions of \high{$X(\lambda_1)z^{(X)}=0$ and $Y(\lambda_1)z^{(Y)}=0$ respectively. In fact, the relations \eqref{N52a} and \eqref{N52b} show that, in the above situation, $z^{(X)} = F(z^2,y^1,y^2)$ and $z^{(Y)} = G(z^1,y^1,y^2)$}. Hence, the introduction of the coordinate system $(z^1,z^2,y^1,y^2)$ constitutes a natural way of getting around the fact that it is impossible to have three or four coinciding spectral parameters in the formalism presented in this paper. The general relationship between the latter and the classical theory of ``straightening'' commuting vector fields using ``eigenfunctions'' (integrals) \cite{Straightening} is currently under investigation.

\subsection{Application to the general heavenly equation}

As a first illustration, we now show explicitly how the general heavenly equation \eqref{N41} may be mapped to the Husain-Park equation \eqref{N50}. Since it is natural to select the parameters
\bela{N53}
  \lambda_1=\lambda_2=\mu_2,\quad \lambda_3=\mu_3,\quad \lambda_4=\mu_4,
\ela
we may make the choice
\bela{N54}
  \phi^1 = \Xi_{x^2},\quad \phi^2 = x^2,\quad \phi^3 = x^3,\quad \phi^4 = x^4.
\ela
If we consider the case $\mu_1\rightarrow\infty$ and $\mu_2=0$ without loss of generality then the Lax pair for the general heavenly equation in the form
\bela{N55}
  (\mu_3 - \mu_4)\Xi_{x^1x^2}\Xi_{x^3x^4} = \mu_3\Xi_{x^2x^3}\Xi_{x^1x^4} - \mu_4\Xi_{x^1x^3}\Xi_{x^2x^4}
\ela
is given by \cite{Schief1996}
\bela{N56}
 \begin{split}
  \Phi_{x^3} & = \frac{1}{(\lambda-\mu_3)\Xi_{x^1x^2}}(\lambda\Xi_{x^1x^3}\Phi_{x^2} - \mu_3\Xi_{x^2x^3}\Phi_{x^1})\\
  \Phi_{x^4} & = \frac{1}{(\lambda-\mu_4)\Xi_{x^1x^2}}(\lambda\Xi_{x^1x^4}\Phi_{x^2} - \mu_4\Xi_{x^2x^4}\Phi_{x^1})
 \end{split}
\ela
so that
\bela{N57}
  A_2 = \del_{x^3} - \frac{\Xi_{x^1x^3}}{\Xi_{x^1x^2}}\del_{x^2},\quad
  A_4 = \del_{x^4} - \frac{\Xi_{x^1x^4}}{\Xi_{x^1x^2}}\del_{x^2}.
\ela
The latter vector fields are divergence free with respect to $f=\Xi_{x^1x^2}$ as pointed out at the end of Section 4. Accordingly, the entries \eqref{N30} of the skew-symmetric matrix $A = \tilde{A}$ (by virtue of $J = \Xi_{x^1x^2}$) read
\bela{N58}
 \bear{c}\dis
  A^{12} = \frac{\Xi_{x^2x^3}\Xi_{x^1x^4} - \Xi_{x^1x^3}\Xi_{x^2x^4}}{\Xi_{x^1x^2}},\quad
  A^{13} = \Xi_{x^2x^4} - \frac{\Xi_{x^1x^4}}{\Xi_{x^1x^2}}\Xi_{x^2x^2}\AS\dis
  A^{41} = \Xi_{x^2x^3} - \frac{\Xi_{x^1x^3}}{\Xi_{x^1x^2}}\Xi_{x^2x^2},\quad
  A^{32} = \frac{\Xi_{x^1x^4}}{\Xi_{x^1x^2}},\quad A^{24} = \frac{\Xi_{x^1x^3}}{\Xi_{x^1x^2}},\quad A^{43} = 1.
 \ear
\ela
On use of the identity
\bela{N59}
  \del_{x^1} = \Xi_{x^1x^2}\del_{y^1},
\ela
we therefore conclude that the connection between the potentials $\Xi$ and $\Theta$ encoded in \eqref{N40} may be formulated as
\bela{N60}
  \Theta_{y^1y^3} = -\frac{1}{2}\frac{\del_{y^1}\Xi_{x^3}}{\mu_3},\quad  
  \Theta_{y^1y^4} = -\frac{1}{2}\frac{\del_{y^1}\Xi_{x^4}}{\mu_4}
\ela
and similar expressions for the remaining mixed derivatives of $\Theta$ except for $\Theta_{y^1y^2}$. Hence, integration leads, without loss of generality, to the first-order relations
\bela{N61}
 \Theta_{y^3} = -\frac{1}{2}\frac{\Xi_{x^3}}{\mu_3},\quad \Theta_{y^4} = -\frac{1}{2}\frac{\Xi_{x^4}}{\mu_4}.
\ela
One may now directly verify that the above pair is compatible modulo the general heavenly equation \eqref{N55} and $\Theta$ indeed satisfies the Husain-Park equation \eqref{N50} with $f_{12} = (\mu_4^{-1}-\mu_3^{-1})/2$.

\subsection{Application to Pleba\'nski's first heavenly equation}

The connection between Pleba\'nski's first heavenly equation and the (elliptic) Husain-Park equation has been established explicitly in \cite{JakimowiczTafel2006}. Here, we demonstrate how this connection may be found algorithmically using our formalism. It is recalled (see Section 2) that the standard Lax pair for the first Pleba\'nski equation
\bela{N62}
 \Omega_{x^1x^3}\Omega_{x^2x^4} - \Omega_{x^2x^3}\Omega_{x^1x^4} = 1
\ela
reads
\bela{N63}
 \begin{split}
  \Phi_{x^3} & = \lambda(\Omega_{x^1x^3}\Phi_{x^2} - \Omega_{x^2x^3}\Phi_{x^1}) = \lambda A_2\Phi\\
  \Phi_{x^4} & = \lambda(\Omega_{x^1x^4}\Phi_{x^2} - \Omega_{x^2x^4}\Phi_{x^1}) = \lambda A_4\Phi.
 \end{split}
\ela
For a non-vanishing spectral parameter $\lambda$, this is equivalent to the pair
\bela{N64}
 \begin{split}
  \Phi_{x^1} & = \lambda^{-1}(\Omega_{x^1x^4}\Phi_{x^3} - \Omega_{x^1x^3}\Phi_{x^4})\\
  \Phi_{x^2} & = \lambda^{-1}(\Omega_{x^2x^4}\Phi_{x^3} - \Omega_{x^2x^3}\Phi_{x^4}).
 \end{split}
\ela

In the following, the most convenient specialisation of the parameters $\lambda_i$ and associated eigenfunctions $\phi^i$ for \highb{$i=1,2$} is given by 
\bela{N65}
 \lambda_1=\lambda_2=0,\quad \phi^1=x^1,\quad \phi^2 = x^2
\ela
\highb{with the remaining eigenfunctions $\phi^3$ and $\phi^4$ corresponding to the parameters $\lambda_3$ and $\lambda_4$ being arbitrary.} Accordingly, we obtain
\bela{N66}
 \bear{c}
  \tilde{A}^{31} = \phi^3_{x^2},\quad \tilde{A}^{41} = \phi^4_{x^2},\quad \tilde{A}^{23} = \phi^3_{x^1},\quad \tilde{A}^{24}=\phi^4_{x^1}\as
  \tilde{A}^{21} = 1,\quad \tilde{A}^{43} = \phi^3_{x^1}\phi^4_{x^2} - \phi^3_{x^2}\phi^4_{x^1}.
 \ear
\ela
Moreover, inversion of the identities
\bela{N67}
  \del_{x^3} = \phi^3_{x^3}\del_{y^3} + \phi^4_{x^3}\del_{y^4},\quad
  \del_{x^4} = \phi^3_{x^4}\del_{y^3} + \phi^4_{x^4}\del_{y^4}
\ela
yields
\bela{N68}
  \del_{y^3} = \frac{\phi^4_{x^4}\del_{x^3} - \phi^4_{x^3}\del_{x^4}}{J},\quad
  \del_{y^4} = \frac{\phi^3_{x^3}\del_{x^4} - \phi^3_{x^4}\del_{x^3}}{J}
\ela
with $J = \phi^3_{x^3}\phi^4_{x^4} - \phi^3_{x^4}\phi^4_{x^3}$ so that the relations \eqref{N66}$_{1,2,3,4}$ become
\bela{N69}
 \tilde{A}^{23}= \frac{J}{\lambda_3}\del_{y^4}\Omega_{x^1},\quad
 \tilde{A}^{31} = \frac{J}{\lambda_3}\del_{y^4}\Omega_{x^2},\quad
 \tilde{A}^{42} = \frac{J}{\lambda_4}\del_{y^3}\Omega_{x^1},\quad
 \tilde{A}^{14}= \frac{J}{\lambda_4}\del_{y^3}\Omega_{x^2} 
\ela
by virtue of the Lax pair \eqref{N64}. Now, since $f=1$ so that \high{$A^{lm}=\tilde{A}^{lm}/J$}, four of the relations \eqref{N40} may be integrated \high{to obtain
\bela{N70}
 \Theta_{y^1} = \frac{1}{2\lambda_3\lambda_4}[\Omega_{x^1} + p(x^1,x^2)],\quad \Theta_{y^2} = \frac{1}{2\lambda_3\lambda_4}[\Omega_{x^2} + q(x^1,x^2)],
\ela
where $p(y^1,y^2)=p(x^1,x^2)$ and $q(y^1,y^2)=q(x^1,x^2)$ are functions of integration. By construction, the above pair must be compatible modulo the first Pleba\'nski equation \eqref{N62} and the Lax pair \eqref{N63} for $\phi^3$ and $\phi^4$ corresponding to the parameters $\lambda_3$ and $\lambda_4$. Indeed, cross-differentiation produces the relation
\bela{N70a}
 \lambda_3\lambda_4(p_{x^2} - q_{x^1}) = \lambda_3 + \lambda_4.
\ela
Hence, without loss of generality, we may choose
\bela{N70b}
  p = \frac{\lambda_3+\lambda_4}{2\lambda_3\lambda_4}x^2,\quad
  q = -\frac{\lambda_3+\lambda_4}{2\lambda_3\lambda_4}x^1.
\ela
Finally, the remaining compatible relation \eqref{N40}$_{i=3,k=4}$, namely
\bela{N71}
  \Theta_{y^3y^4} = \frac{1}{2J(\lambda_3-\lambda_4)},
\ela
guarantees that $\Theta$ is a solution of the Husain-Park equation \highb{\eqref{N50}} for $f_{12}=(\lambda_3-\lambda_4)/2\lambda_3^2\lambda_4^2$.

It is evident that the choice $\lambda_4=-\lambda_3$ is privileged since, in this case, the pair \eqref{N70} simplifies to
\bela{N71a}
 \Theta_{y^1} = -\frac{1}{2\lambda_3^2}\Omega_{x^1},\quad \Theta_{y^2} = -\frac{1}{2\lambda_3^2}\Omega_{x^2}\highb{\bullet}
\ela
The latter represents the analogue of the relations derived in \cite{JakimowiczTafel2006} for the elliptic Husain-Park equation. Indeed, if one sets aside the ``normalisation'' \eqref{N71} then $\Theta$ defined by the compatible pair \eqref{N71a} constitutes a solution of the Husain-Park equation \eqref{N50} modulo a suitable gauge transformation of the form $\Theta\rightarrow\Theta + f_{34}(y^3,y^4)$.} 

\section{Two pairs of coinciding spectral parameters. Pleba\'nski's first heavenly equation}

Here, we consider the case of {\em two} pairs of coinciding parameters, say, $\lambda_1=\lambda_2$ and $\lambda_3=\lambda_4$. This case can be dealt with in the same manner as before, leading to the parametrisation
\bela{N72}
  \omega_{ik} = (\lambda_k - \lambda_i)\Theta_{y^iy^k},\quad \omega_{12} = \omega_{12}(y^1,y^2),\quad \omega_{34} = \omega_{34}(y^3,y^4),
\ela
wherein $(i,k)\not\in\{(1,2),(2,1),(3,4),(4,3)\}$. The Pfaffian condition \eqref{N22} then becomes
\bela{N73}
  f_{12}f_{34} + \Theta_{y^2y^3}\Theta_{y^1y^4} - \Theta_{y^3y^1}\Theta_{y^2y^4} = 0
\ela
with
\bela{N74}
  f_{12} = \frac{\omega_{12}(y^1,y^2)}{\lambda_1-\lambda_3},\quad  f_{34} = \frac{\omega_{34}(y^3,y^4)}{\lambda_1-\lambda_3},
\ela
which, on application of a suitable coordinate transformation of the form $(y^1,y^2)\rightarrow\fb_{12}(y^1,y^2)$ and $(y^3,y^4)\rightarrow\fb_{34}(y^3,y^4)$, constitutes Pleba\'nski's first heavenly equation corresponding to $f_{12}=f_{34}=1$. Once again, it is remarked in passing that a confluence limit of the type
\bela{N75}
  \Theta\rightarrow \Theta + \frac{g_{12}(y^1,y^2)}{\epsilon} +  \frac{g_{34}(y^3,y^4)}{\epsilon},\quad \lambda_2 = \lambda_1 + \epsilon,\quad\lambda_4 = \lambda_3 + \epsilon, \quad \epsilon\rightarrow0
\ela
directly reduces the general heavenly equation \eqref{N23} to the first heavenly equation in the ``gauge-invariant'' form \eqref{N73}.

As in the previous section, one may now map any of the known heavenly equations to Pleba\'nski's first heavenly equation. For instance, application of our formalism to the Husain-Park equation results in the ``inverse'' of the transformation from the first heavenly equation to the Husain-Park equation (cf.\ \cite{JakimowiczTafel2006}) derived in the previous section. Here, we focus on the application to Pleba\'nski's second heavenly equation and the general heavenly equation.

\subsection{Application to Pleba\'nski's second heavenly equation}

The classical link \cite{Plebanski1975,DoubrovFerapontov2010} between Pleba\'nski's second heavenly equation
\bela{A1}
  \Lambda_{x^1x^3} + \Lambda_{x^2x^4} + \Lambda_{x^1x^1}\Lambda_{x^2x^2} - \Lambda_{x^1x^2}^2 = 0
\ela
and the first heavenly equation may be formulated in terms of the two-form
\bela{A2}
 \highb{\hat{\Omega}} = (dx^1 - \Lambda_{x^2x^2}dx^3 + \Lambda_{x^1x^2}dx^4)\wedge(\high{dx^2 +  \Lambda_{x^1x^2}dx^3 - \Lambda_{x^1x^1}dx^4})
\ela
which has the properties
\bela{A3}
  \highb{d\hat{\Omega} = 0,\quad \hat{\Omega}\wedge\hat{\Omega} = 0.}
\ela
Accordingly, Darboux's theorem guarantees the existence of functions $y^1$ and $y^2$ such that
\bela{A4}
  \highb{\hat{\Omega} = dy^1\wedge dy^2.}
\ela
Comparison with \eqref{A2} shows that there exist expansions of the form
\bela{A5}
 \begin{split}
  dx^1 - \Lambda_{x^2x^2}dx^3 + \Lambda_{x^1x^2}dx^4 & = u_{13}dy^1 + u_{23}dy^2\\
  \high{dx^2 + \Lambda_{x^1x^2}dx^3 - \Lambda_{x^1x^1}dx^4} & = u_{14}dy^1 + u_{24}dy^2
 \end{split}
\ela
for some functions $u_{13},u_{24},u_{23}$ and $u_{14}$ subject to
\bela{A6}
  u_{13}u_{24} - u_{23}u_{14} = 1.
\ela
If we now regard $y^1,y^2$ and $y^3=x^3,\,y^4=x^4$ as independent variables then the expansions \eqref{A5} imply that $x^1_{y^4} = x^2_{y^3}$ so that there exists a potential $\Theta$ defined according to
\bela{A7}
  x^1 = \Theta_{y^3},\quad x^2 = \Theta_{y^4}.
\ela
Hence, \eqref{A5} gives rise to the parametrisation
\bela{A8}
  u_{ik} = \Theta_{y^iy^k}
\ela
which, in turn, reveals that the algebraic relation \eqref{A6} encodes Pleba\'nski's first heavenly equation
\bela{A9}
  \Theta_{y^1y^3}\Theta_{y^2y^4} - \Theta_{y^2y^3}\Theta_{y^1y^4} = 1.
\ela

The connection with the present formalism is now uncovered by investigating the nature of the coordinates $y^i$. Thus, if we solve \eqref{A5} for $dy^1$ and $dy^2$ then we obtain
\bela{A10}
 \begin{split}
  dy^1 & = u_{24}(dx^1 - \Lambda_{x^2x^2}dx^3 + \Lambda_{x^1x^2}dx^4) - u_{23}(dx^2 + \Lambda_{x^1x^2}dx^3 - \Lambda_{x^1x^1}dx^4)\\
  dy^2 & = u_{13}(\high{dx^2} + \Lambda_{x^1x^2}dx^3 - \Lambda_{x^1x^1}dx^4) - u_{14}(\high{dx^1} - \Lambda_{x^2x^2}dx^3 + \Lambda_{x^1x^2}dx^4).
 \end{split}
\ela
The latter implies that
\bela{A11}
  \highb{u_{13} = y^2_{x^2},\quad u_{24} = y^1_{x^1}},\quad u_{23} = -y^1_{x^2},\quad \high{u_{14} = -y^2_{x^1}}
\ela
so that \eqref{A10} reduces to
\bela{A12}
 \begin{split}
  y^i_{x^3} = \Lambda_{x^1x^2}y^i_{x^2} - \Lambda_{x^2x^2}y^i_{x^1},\\
  y^i_{x^4} = \Lambda_{x^1x^2}y^i_{x^1} - \Lambda_{x^1x^1}y^i_{x^2},
 \end{split}\quad i=1,2.
\ela
Hence, comparison with the Lax pair \cite{DoubrovFerapontov2010}
\bela{A13}
 \begin{aligned}
  X(\lambda)\Phi &= 0,\quad &X(\lambda) &= \del_{x^3} - \Lambda_{x^1x^2}\del_{x^2} + \Lambda_{x^2x^2}\del_{x^1} - \lambda\del_{x^2}\\ 
  Y(\lambda)\Phi &= 0,\quad &Y(\lambda) &= \del_{x^4} - \Lambda_{x^1x^2}\del_{x^1} + \Lambda_{x^1x^1}\del_{x^2} + \lambda\del_{x^1}
 \end{aligned}
\ela
for the second heavenly equation \eqref{A1} shows that the pairs $y^1,y^2$ and $y^3,y^4$ constitute eigenfunctions for $\lambda = 0$ and $\lambda\rightarrow\infty$ respectively. 

It is evident that our formalism is directly applicable even if one or two parameters $\lambda_i$ vanish or, due to the symmetry $\lambda\rightarrow\lambda^{-1}$, one or two parameters tend to infinity. If vanishing and infinite parameters are simultaneously present then the situation is more subtle but it is easy to verify that the formalism also applies {\em mutatis mutandis} in this case. \highb{In the current situation, the second-order relations between the potential $\Theta$ and the original quantities associated with the second heavenly equation are obtained by eliminating $u_{ik}$ between \eqref{A8} and \eqref{A11}}. Accordingly, the link between the two heavenly equations presented in Pleba\'nski's pioneering work \cite{Plebanski1975} may be regarded as a particular application of the scheme presented here.

\subsection{Application to the general heavenly equation}

Since the connection between the general heavenly equation and the first Pleba\'nski equation constitutes the basis of Section 11 which, in turn, \high{justifies} the reasoning employed in Section 10, we now derive the differential relations between the potentials satisfying those two equations. Thus, in order to map the general heavenly equation \eqref{N55} to the first heavenly equation, it is convenient to choose the eigenfunctions
\bela{A14}
 \phi^1 = \Xi_{x^3},\quad \phi^2 = x^3,\quad \phi^3 = \Xi_{x^4},\quad \phi^4 = x^4
\ela
corresponding to the parameters
\bela{A15}
  \lambda_1 = \lambda_2 = \mu_3,\quad \lambda_3 = \lambda_4 = \mu_4.
\ela
Then, evaluation of \eqref{N30} for the vector fields $A_2$ and $A_4$ given by \eqref{N57} results in
\bela{A16}
 \begin{split}
  \tilde{A}^{13} & = \Big(\Xi_{x^3x^4} - \frac{\Xi_{x^1x^4}}{\Xi_{x^1x^2}}\Xi_{x^2x^3}\Big)
                    \Big(\Xi_{x^3x^4} - \frac{\Xi_{x^1x^3}}{\Xi_{x^1x^2}}\Xi_{x^2x^4}\Big)\\
              & -  \Big(\Xi_{x^3x^3} - \frac{\Xi_{x^1x^3}}{\Xi_{x^1x^2}}\Xi_{x^2x^3}\Big)
                    \Big(\Xi_{x^4x^4} - \frac{\Xi_{x^1x^4}}{\Xi_{x^1x^2}}\Xi_{x^2x^4}\Big)\\
  \tilde{A}^{41} & = \Big(\Xi_{x^3x^3} - \frac{\Xi_{x^1x^3}}{\Xi_{x^1x^2}}\Xi_{x^2x^3}\Big),\quad \tilde{A}^{42} = 1\\
  \tilde{A}^{32} & = \Big(\Xi_{x^4x^4} - \frac{\Xi_{x^1x^4}}{\Xi_{x^1x^2}}\Xi_{x^2x^4}\Big).
 \end{split}
\ela
By construction, the mixed derivatives
\bela{A17}
 \begin{aligned}
\Theta_{y^1y^3} &= \frac{1}{2}\frac{A^{42}}{\mu_4-\mu_3},\quad &\Theta_{y^2y^4} &= \frac{1}{2}\frac{A^{31}}{\mu_4-\mu_3}\\
\Theta_{y^2y^3} &= \frac{1}{2}\frac{A^{14}}{\mu_4-\mu_3},\quad &\Theta_{y^1y^4} &= \frac{1}{2}\frac{A^{23}}{\mu_4-\mu_3},
 \end{aligned}
\ela
where
\bela{A18}
  A^{ik} = \frac{f}{J}\tilde{A}^{ik},\quad f = \Xi_{x^1x^2},\quad J = \Xi_{x^2x^3}\Xi_{x^1x^4} - \Xi_{x^1x^3}\Xi_{x^2x^4},
\ela
are compatible modulo the general heavenly equation \eqref{N55} and the potential $\Theta$ obeys Pleba\'nski's first heavenly equation in the form
\bela{A19}
  \Theta_{y^1y^3}\Theta_{y^2y^4} - \Theta_{y^2y^3}\Theta_{y^1y^4} = \sigma,\qquad \sigma = - \frac{1}{4}\frac{\mu_3\mu_4}{(\mu_4-\mu_3)^4}
\ela
as required. 

As indicated above, the relations \eqref{A17} provide a key link in the remaining discussion. In order to motivate the content of Sections 10 and 11, we first review an important known fact in the context of the general heavenly equation. \highb{Before we do so, we may summarise the key results of Sections 4, 7 and 8 in the following table.}
\begin{center}
    \begin{tabular}{| p{3.5cm} | p{4cm} |}
    \hline
    \# of pairs of coinciding parameters $\lambda_i$ & Canonical form of the self-dual Einstein equations\\ \hline\hline
      \qquad\qquad\, 0& General heavenly equation\\ \hline
      \qquad\qquad\, 1 & Husain-Park equation\\
    \hline
   \qquad\qquad\, 2 & First heavenly equation\\
   \hline
    \end{tabular}
\end{center}

\section{The dispersionless Hirota equation. Einstein-Weyl geometry}

The Jones-Tod procedure \cite{JonesTod1985} provides a connection between \high{four-dimensional spacetimes with anti-self-dual Weyl tensor and a conformal Killing vector} and three-dimensional Einstein-Weyl geometries. In particular, in \cite{DunajskiKrynski2014}, it has been shown how the dispersionless Hirota equation (see, e.g., \cite{Krynski2018}) simultaneously gives rise to Einstein-Weyl geometries and a particular class of \high{(anti-)self-dual Einstein spaces}. Here, we discuss this observation in connection with the general heavenly equation with a view to the generalisation presented in Section~10.  

The metric of self-dual Einstein spaces governed by the general heavenly equation
\bela{Z1}
 (\mu_3 - \mu_4)\Xi_{x^1x^2}\Xi_{x^3x^4} = \mu_3\Xi_{x^2x^3}\Xi_{x^1x^4} - \mu_4\Xi_{x^1x^3}\Xi_{x^2x^4}
\ela
with associated Lax representation
\bela{Z2}
 \begin{split}
  \Phi_{x^3} & = \frac{1}{(\lambda-\mu_3)\Xi_{x^1x^2}}(\lambda\Xi_{x^1x^3}\Phi_{x^2} - \mu_3\Xi_{x^2x^3}\Phi_{x^1})\\
  \Phi_{x^4} & = \frac{1}{(\lambda-\mu_4)\Xi_{x^1x^2}}(\lambda\Xi_{x^1x^4}\Phi_{x^2} - \mu_4\Xi_{x^2x^4}\Phi_{x^1})
 \end{split}
\ela
is known to be given by \cite{Schief1996}
\bela{Z3}
  g = q^{-1}[\Xi_{x^1x^2}\Xi_{x^1x^3}\Xi_{x^1x^4}{(dx^1)}^2 + \Xi_{x^1x^2}(\Xi_{x^1x^3}\Xi_{x^2x^4} + \Xi_{x^1x^4}\Xi_{x^2x^3}) dx^1dx^2 + \cdots],
\ela
where $q=f_{1234}$ and
\bela{Z4}
 f_{iklm} = \Xi_{x^ix^k}\Xi_{x^lx^m} - \Xi_{x^ix^l}\Xi_{x^kx^m}
\ela
for distinct indices $i,k,l,m$. Indeed, one may directly verify that $R_{ik}=0$ modulo the general heavenly equation \eqref{Z1}. It is noted that there exist only three essentially different quantities $f_{iklm}$ and their ratios are constant due to the structure of the general heavenly equation. Hence, up to an irrelevant constant factor, the metric \eqref{Z3} is completely symmetric in the indices.

We now single out a coordinate, say, $x^4$ and split the metric \eqref{Z3} into a ``three-dimensional'' metric and a ``perfect square'' according to
\bela{Z5}
  g = -\frac{f_{1234}}{4}h + \frac{\Xi_{x^1x^4}\Xi_{x^2x^4}\Xi_{x^3x^4}}{f_{1234}}(dx^4 + \eta)^2,
\ela
where
\bela{Z6}
 \begin{split}
  h & = \frac{\Xi_{x^1x^4}}{\Xi_{x^2x^4}\Xi_{x^3x^4}}{(dx^1)}^2
     + \alpha^2\frac{\Xi_{x^2x^4}}{\Xi_{x^1x^4}\Xi_{x^3x^4}}{(dx^2)}^2
     + \beta^2\frac{\Xi_{x^3x^4}}{\Xi_{x^1x^4}\Xi_{x^2x^4}}{(dx^3)}^2\\
   & + \frac{2\alpha}{\Xi_{x^3x^4}}dx^1dx^2
      + \frac{2\beta}{\Xi_{x^2x^4}}dx^1dx^3
      - \frac{2\alpha\beta}{\Xi_{x^1x^4}}dx^2dx^3
 \end{split}
\ela
with the constants
\bela{Z7}
 \alpha = \frac{f_{1243}}{f_{1234}},\quad \beta = \frac{f_{1432}}{f_{1234}}
\ela
and
\bela{Z9}
 \eta = \frac{1}{2}\left(\frac{\Xi_{x^1x^2}}{\Xi_{x^2x^4}}+\frac{\Xi_{x^1x^3}}{\Xi_{x^3x^4}}\right)dx^1
        + \frac{1}{2}\left(\frac{\Xi_{x^1x^2}}{\Xi_{x^1x^4}}+\frac{\Xi_{x^2x^3}}{\Xi_{x^3x^4}}\right)dx^2
        + \frac{1}{2}\left(\frac{\Xi_{x^1x^3}}{\Xi_{x^1x^4}}+\frac{\Xi_{x^2x^3}}{\Xi_{x^2x^4}}\right)dx^3. 
\ela
Furthermore, we consider the admissible reduction
\bela{Z10}
  \Xi = f(x^4)\omega(x^1,x^2,x^3)
\ela
which specialises the general heavenly equation \eqref{Z1} to the dispersionless Hirota equation
\bela{Z11}
  (\mu_3 - \mu_4)\omega_{x^1x^2}\omega_{x^3} - \mu_3\omega_{x^2x^3}\omega_{x^1} + \mu_4\omega_{x^1x^3}\omega_{x^2} = 0.
\ela
Since, the dependence of the metric \eqref{Z3} on $x^4$ is now merely encoded in the overall factor \high{$f(x^4)$}, it is evident that the reduction \eqref{Z10} leads to self-dual Einstein spaces admitting a homothetic Killing vector. Moreover, we obtain
\bela{Z12}
 \begin{split}
  h & \sim \frac{\omega_{x^1}}{\omega_{x^2}\omega_{x^3}}{(dx^1)}^2
     + \alpha^2\frac{\omega_{x^2}}{\omega_{x^1}\omega_{x^3}}{(dx^2)}^2
     + \beta^2\frac{\omega_{x^3}}{\omega_{x^1}\omega_{x^2}}{(dx^3)}^2\\
   & + \frac{2\alpha}{\omega_{x^3}}dx^1dx^2
      + \frac{2\beta}{\omega_{x^2}}dx^1dx^3
      - \frac{2\alpha\beta}{\omega_{x^1}}dx^2dx^3,
 \end{split}
\ela
up to an irrelevant factor depending on $x^4$, which is precisely the metric governing the Einstein-Weyl geometry associated with the dispersionless Hirota equation \cite{DunajskiKrynski2014}.

\section{A dispersionless Hirota system. Self-dual Einstein\newline spaces not admitting conformal Killing vectors}

A non-trivial reduction of the general heavenly equation is obtained by matching its scaling symmetry $\Xi\del_{\Xi}$ with the symmetry $\Phi\del_{\Xi}$. The fact that any eigenfunction of the general heavenly equation constitutes a symmetry of the general heavenly equation was first observed in \cite{Sergyeyev2017} and extends to its 4+4-dimensional version (TED equation) \cite{KonopelchenkoSchief2019} as further discussed in Section 12. Thus, if we set $\Phi=\Xi$ then the Lax pair \eqref{Z2} gives rise to the pair of dispersionless Hirota equations
\bela{Z13}
 \begin{split}
  (\lambda-\mu_3)\Xi_{x^3}\Xi_{x^1x^2} - \lambda\Xi_{x^2}\Xi_{x^1x^3} + \mu_3\Xi_{x^1}\Xi_{x^2x^3} &= 0\\
  (\lambda-\mu_4)\Xi_{x^4}\Xi_{x^1x^2} - \lambda\Xi_{x^2}\Xi_{x^1x^4} + \mu_4\Xi_{x^1}\Xi_{x^2x^4} & = 0
 \end{split}
\ela
which is, by construction, compatible with the general heavenly equation \eqref{Z1}. Moreover, the pair of dispersionless Hirota equations
\bela{Z14}
 \begin{split}
 (\mu_3-\mu_4)\Xi_{x^1}\Xi_{x^3x^4} + (\lambda-\mu_3)\Xi_{x^3}\Xi_{x^1x^4} - (\lambda-\mu_4)\Xi_{x^4}\Xi_{x^1x^3} & = 0\\
 \lambda(\mu_3-\mu_4)\Xi_{x^2}\Xi_{x^3x^4} + \mu_4(\lambda-\mu_3)\Xi_{x^3}\Xi_{x^2x^4} - \mu_3(\lambda-\mu_4)\Xi_{x^4}\Xi_{x^2x^3} & = 0
 \end{split}
\ela
is an algebraic consequence of the dispersionless Hirota equations \eqref{Z13} and the general heavenly equation. In fact, any three of the four dispersionless Hirota equations imply the fourth and the general heavenly equation. 
\high{It is noted that the dispersionless Hirota system is completely symmetric in its indices if the symmetry in the parameters $\mu_1=\infty$, $\mu_2=0$, $\mu_3$, $\mu_4$ is restored by means of a suitable fractional linear transformation of the parameters, leading to the fully symmetric form \highb{\eqref{N41}} of the general heavenly equation.} The compatibility of copies of dispersionless Hirota equations was first observed in \cite{Krynski2018} and is a direct consequence of the multi-dimensional consistency of the general heavenly equation \cite{Bogdanov2015} or, more generally, the 4+4-dimensional TED equation \cite{KonopelchenkoSchief2019}. It is also emphasised that the above dispersionless Hirota system is invariant under $\Xi\rightarrow F(\Xi)$. Even though this invariance may be proven directly, it is a consequence of the fact that any function of an eigenfunction of the general heavenly equation constitutes another eigenfunction. 

\begin{remark}
If we impose the constraint \eqref{Z10} and make the choice $\lambda=\mu_4$ then the dispersionless Hirota system \eqref{Z13}, \eqref{Z14} reduces to the dispersionless Hirota equation \eqref{Z11}. Hence, remarkably, the Einstein-Weyl geometry associated with the dispersionless Hirota equation is captured as a special case by the class of self-dual Einstein spaces governed by the eigenfunction symmetry reduction leading to the dispersionless Hirota system. In fact, the dispersionless Hirota system specialises to the dispersionless Hirota equation if any of the four constraints
\bela{Z14a}
  \Xi = f(x^i)\omega(x^k,x^l,x^m),\quad \lambda = \mu_i,
\ela
where the indices $i,k,l,m$ are distinct, is imposed. Accordingly, the Einstein-Weyl geometry discussed in Section 9 is encoded in four different ways in the self-dual Einstein spaces examined in this section so that, in this sense, the algebraic multi-dimensional consistency of the dispersionless Hirota equation has its geometric counterpart in the ``multi-dimensional consistency'' of its associated Einstein-Weyl geometry. The exact nature of this geometric property
is the subject of ongoing research. 
\end{remark}

\subsection{The symmetry \boldmath $\Xi\rightarrow F(\Xi)$}

In the preceding, it has been demonstrated that the four-dimensional general heavenly equation admits a decomposition into four compatible three-dimensional dispersionless Hirota equations (of which only three are independent). It is therefore natural to inquire as to whether the solutions of the general heavenly equation obtained in this manner are genuinely four-dimensional in the sense that the corresponding self-dual Einstein spaces do not admit conformal Killing symmetries. Furthermore, it is desirable to show that the symmetry $\Xi\rightarrow F(\Xi)$ really acts non-trivially on the metric. At first glance, this appears to be likely due to the appearance of the function $F$ in the transformed metric. We begin by addressing the latter problem.

We first observe that
\bela{Z15}
 \Xi = z^1 + f(z^2),\quad z^1 = \sum_{i=1}^4 x^i,\quad z^2 = \sum_{i=1}^4\alpha_ix^i
\ela
constitutes a trivial solution of the general heavenly equation \eqref{Z1} and, in order for it to satisfy the pair \eqref{Z13} of dispersionless Hirota equations, the constants $\alpha_i$ and $\mu_k$ must be related by
\bela{Z16}
  \mu_k = \lambda\frac{\alpha_1(\alpha_2-\alpha_k)}{\alpha_2(\alpha_1-\alpha_k)}, \quad k=3,4.
\ela
However, this solution does not correspond to a viable metric since $f_{1234}=0$. This may be rectified by boosting the solution using the symmetry $\Xi\rightarrow F(\Xi)$ to obtain
\bela{Z17}
  \Xi = F(z^1 + f(z^2)).
\ela
\highb{Indeed, the condition for non-vanishing $f_{1234}$ is given by \mbox{$(\alpha_2-\alpha_3)(\alpha_1-\alpha_4)F^{\prime\prime}(u)f^{\prime\prime}(v)\neq0$}.} For example, the solution
\bela{Z18}
  \Xi = e^{z^1}\cosh z^2,
\ela
corresponding to $F(u)=e^u,\,f(v)=\ln\cosh v$, falls into this category, \highb{provided that \mbox{$\alpha_2\neq\alpha_3$} and $\alpha_1=\alpha_4$ being excluded due to \eqref{Z16}}. This solution generates a flat metric if, for instance, $\alpha_1=-\alpha_2=1$. More generally, for this particular choice of the function $f$ and the constants $\alpha_i$ \highb{but arbitrary function $F$ so that $\Xi = F(z^1 + \ln\cosh z^2)$, the components of the Riemann tensor which are not identically zero turn out to be proportional to}
\bela{Z19}
 R_{iklm} \sim \highb{\alpha_3\alpha_4}\Big({[F^{\prime}(u)]}^3F^{\prime\prime}(u)F^{\prime\prime\prime\prime}(u) - 3{[F^{\prime}(u)]}^3{[F^{\prime\prime\prime}(u)]}^2 + 2{[F^{\prime\prime}(u)]}^5\Big),
\ela
where $u = z^1 + \ln\cosh z^2$. Thus, for a four-parameter family of functions $F$, which includes $F(u)=e^u$, the metric is flat but, generically, the invariance $\Xi\rightarrow F(\Xi)$ maps the flat metric associated with the solution \eqref{Z18} to a non-flat metric. \highb{Here, we exclude the case $\alpha_3\alpha_4=0$}. This proves that this invariance of the dispersionless Hirota system is non-trivial at the level of the metric.

\subsection{Non-existence of conformal Killing vectors}

In this section, we prove the following theorem.

\begin{theorem}\label{killing}
 The generic metric of self-dual Einstein spaces governed by the dispersionless Hirota system \eqref{Z13}, \eqref{Z14} does not admit conformal Killing vectors (including homothetic Killing vectors and Killing vectors).
\end{theorem}

It turns out convenient to examine the above problem in the setting of Pleba\'nski's first heavenly equation
\bela{Z20}
  \Theta_{y^1y^3}\Theta_{y^2y^4} - \Theta_{y^2y^3}\Theta_{y^1y^4} = 1
\ela
with associated Lax pair 
\bela{Z21}
 \begin{split}
  \Phi_{y^3} & = \highb{\bar{\lambda}}(\Theta_{y^1y^3}\Phi_{y^2} - \Theta_{y^2y^3}\Phi_{y^1})\\
  \Phi_{y^4} & = \highb{\bar{\lambda}}(\Theta_{y^1y^4}\Phi_{y^2} - \Theta_{y^2y^4}\Phi_{y^1}).
 \end{split}
\ela
In Section 8, it has been demonstrated how the current formalism may be used to establish the link between the general and first heavenly equations. \highb{Extension of this link to the Lax pairs is readily shown to lead to the above standard Lax pair with $\bar{\lambda}$ being related to $\lambda$ by a fractional linear transformation (cf.\ \eqref{A28})}. In Section 11, this link is exploited to reveal how the decomposition of the general heavenly equation into the dispersionless Hirota system acts on the first heavenly equation. Remarkably, this decomposition corresponds to the assumption that
\bela{Z22}
 \Phi = \Theta - y^1\Theta_{y^1} - y^3\Theta_{y^3}
\ela
constitutes an eigenfunction of the first heavenly equation. Insertion of $\Phi$ as given by \eqref{Z22} into the Lax pair \eqref{Z21} leads to two constraints on $\Theta$ which are compatible with the Pleba\'nski equation \eqref{Z20}. This is discussed in Section 11 in more detail. In order to understand the nature of these constraints, we first briefly examine the symmetries of the first heavenly equation.

\subsubsection{Symmetries of the first heavenly equation}

Here, we consider symmetries of the first Pleba\'nski equation, that is, flows
\bela{Z23}
  \Theta_s = \Delta
\ela
which leave invariant the first heavenly equation \eqref{Z20}. Thus, differentiation of the latter with respect to the symmetry parameter $s$ produces the linear partial differential equation
\bela{Z24}
  \Theta_{y^2y^4}\Delta_{y^1y^3} + \Theta_{y^1y^3}\Delta_{y^2y^4} -  \high{\Theta_{y^1y^4}\Delta_{y^2y^3} - \Theta_{y^2y^3}\Delta_{y^1y^4} = 0}
\ela
for the function $\Delta$. Any solution of this differential equation gives rise to a symmetry of the first heavenly equation. For instance, it is evident that the latter is invariant under the scaling $\Theta\rightarrow\epsilon\Theta$, $y^i\rightarrow\epsilon^2y^i$ for any fixed $i$. This corresponds to
\bela{Z25}
  \Delta = \Theta - 2y^i\Theta_{y^i}.
\ela
One may directly verify that this is indeed a solution of the differential equation \eqref{Z24}. Another symmetry is given by
\bela{Z26}
  \Delta = \Phi
\ela
which may be seen immediately by cross-differentiating the Lax pair \eqref{Z21}, leading to
\bela{Z27}
 \Theta_{y^2y^4}\Phi_{y^1y^3} + \Theta_{y^1y^3}\Phi_{y^2y^4} -  \high{\Theta_{y^1y^4}\Phi_{y^2y^3} - \Theta_{y^2y^3}\Phi_{y^1y^4} = 0},
\ela
provided that \highb{$\bar{\lambda}\neq0$}. This symmetry has been recorded in \cite{MalykhNutkuSheftel2003} in connection with the notion of partner symmetries and is also a direct consequence of the eigenfunction symmetry of the 4+4-dimensional TED equation. We conclude that the admissible constraint \eqref{Z22} corresponds to a symmetry reduction of the first heavenly equation associated with an appropriate linear combination of the symmetries \eqref{Z25} and \eqref{Z26}.

Less obvious symmetries are provided by conservations laws associated with the first heavenly equation. Thus, the compatibility condition associated with the pair
\bela{Z28}
 \begin{split}
  \phi^a_{y^1} &= \Theta_{y^3}\Theta_{y^1y^4} - \Theta_{y^4}\Theta_{y^1y^3} + 2y^2\\
  \phi^a_{y^2} &= \Theta_{y^3}\Theta_{y^2y^4} - \Theta_{y^4}\Theta_{y^2y^3}
 \end{split}
\ela
is precisely the first heavenly equation \eqref{Z20} so that the existence of the potential $\phi^a$ is guaranteed. For reasons of symmetry, the existence of a potential $\phi^b$ defined according to
\bela{Z29}
 \begin{split}
  \phi^b_{y^3} &= \Theta_{y^2}\Theta_{y^1y^3} - \Theta_{y^1}\Theta_{y^2y^3} - 2y^4\\
  \phi^b_{y^4} &= \Theta_{y^2}\Theta_{y^1y^4} - \Theta_{y^1}\Theta_{y^2y^4}
 \end{split}
\ela
is likewise guaranteed. Differentiation of \eqref{Z28} and \eqref{Z29} with respect to $y^3,y^4$ and $y^1,y^2$ respectively then shows that the quantities
\bela{Z30}
  \Delta^a = \phi^a,\quad \Delta^b = \phi^b
\ela
satisfy the differential equation \eqref{Z24} and therefore constitute symmetries of the first heavenly equation.

\subsubsection{The homothetic Killing equations}

We will now investigate the existence of homothetic Killing vectors of the class of metrics associated with solutions of the dispersionless Hirota system and then address the (non-)admittance of more general proper conformal Killing vectors. It is recalled (Section 2) that, in terms of solutions of the first heavenly equation \eqref{Z20}, the metric of self-dual Einstein spaces adopts the form
\bela{Z31}
  g = 2\Theta_{y^1y^3}dy^1dy^3 + 2\Theta_{y^2y^3}dy^2dy^3 + 2\Theta_{y^1y^4}dy^1dy^4 + 2\Theta_{y^2y^4}dy^2dy^4.
\ela
A vector field $V^k$ constitutes a homothetic Killing vector field if its covariant representation $V_k= g_{kl}V^l$ satisfies Killing's equations \cite{SKMHH2003} given by
\bela{Z32}
  \nabla_{(i}V_{k)} = \chi g_{ik},
\ela
wherein $\nabla_i$ denotes the covariant derivative, the round brackets indicate the standard symmetrised sum and $\chi$ is a constant. If the latter vanishes then $V^k$ is a Killing vector giving rise to an isometry, while $\chi\neq0$ corresponds to a proper homothetic Killing vector. Based on a spinor formulation, it has been demonstrated in \cite{BoyerFinley1982} that the above Killing equations may be completely resolved. In fact, in the present formulation, one may directly integrate this set of linear partial differential equations to obtain
\bela{Z33}
 \begin{split}
  V_1 & = c^a(\Theta_{y^4}\Theta_{y^1y^3} - \Theta_{y^3}\Theta_{y^1y^4}) - a(y^3,y^4)\Theta_{y^1y^3} - b(y^3,y^4)\Theta_{y^1y^4}\\
  V_2 & = c^a(\Theta_{y^4}\Theta_{y^2y^3} - \Theta_{y^3}\Theta_{y^2y^4}) - a(y^3,y^4)\Theta_{y^2y^3} - b(y^3,y^4)\Theta_{y^2y^4}\\
  V_3 & = c^b(\Theta_{y^2}\Theta_{y^1y^3} - \Theta_{y^1}\Theta_{y^2y^3}) + c(y^1,y^2)\Theta_{y^1y^3} + d(y^1,y^2)\Theta_{y^2y^3}\\
  V_4 & = c^b(\Theta_{y^2}\Theta_{y^1y^4} - \Theta_{y^1}\Theta_{y^2y^4}) +  c(y^1,y^2)\Theta_{y^1y^4} + d(y^1,y^2)\Theta_{y^2y^4},
 \end{split}
\ela
where $c^a,c^b$ are constants, so that the components of the homothetic Killing vector are given by
\bela{Z33a}
 \begin{aligned}
 V^1 &= \high{c^b\Theta_{y^2}} + c(y^1,y^2),\,\, &V^2 &= -c^b\Theta_{y^1} + d(y^1,y^2)\\
 V^3 &= c^a\Theta_{y^4} - a(y^3,y^4),\,\, &V^4 &= -c^a\Theta_{y^3} - b(y^3,y^4).
 \end{aligned}
\ela
The Killing equations reduce to the ``master equation'' (in the terminology of \cite{BoyerFinley1982})
\bela{Z34}
  \Delta^h = 0,
\ela
where
\bela{Z35}
  \Delta^h = c^a\phi^a + a(y^3,y^4)\Theta_{y^3} + b(y^3,y^4)\Theta_{y^4} + 2\chi\Theta - c(y^1,y^2)\Theta_{y^1} - d(y^1,y^2)\Theta_{y^2} - c^b\phi^b
\ela
and $\phi^a$ and $\phi^b$ are the potentials defined by \eqref{Z28} and \eqref{Z29} respectively. Moreover, the functions $a,b,c,d$ are constrained by the linear equation
\bela{Z36}
  a_{y^3} + b_{y^4} + 4\chi - c_{y^1} - d_{y^2} = 0.
\ela
\highb{In fact, the latter constraint constitutes the trace $\nabla^iV_i = 4\chi$ of the Killing equations \eqref{Z32}.} It may be solved explicitly since the dependence of the pairs $a,b$ and $c,d$ on different variables implies the separation
\bela{Z37}
  a_{y^3} + b_{y^4} + \high{2\chi - \mu} = 0,\quad
  c_{y^1} + \high{d_{y^2} - 2\chi - \mu} = 0,
\ela
where $\mu$ is an arbitrary constant. In summary, the first heavenly equation \eqref{Z20} admits a homothetic Killing vector if and only if it is constrained by the (non-local) condition \eqref{Z34}.

It turns out that the constraint \eqref{Z34} is compatible with the first heavenly equation \eqref{Z20}. For instance, if $c^a=c^b=0$ then \highb{$
\Delta^h=0$} constitutes a first-order constraint which may be shown to lead to a three-dimensional reduction of the first heavenly equation. If $c^ac^b\neq0$ then the necessary conditions $\Delta^h_{y^1y^3}=\Delta^h_{y^1y^4}=\Delta^h_{y^2y^3}=\Delta^h_{y^2y^4}=0$ lead to four third-order differential constraints which are compatible with the first heavenly equation. Conversely, if those four constraints are satisfied then the functions of integration in the definitions \eqref{Z28} and \eqref{Z29} of the potentials $\phi^a$ and $\phi^b$ respectively may be chosen in such a manner that the non-local condition $\Delta^h=0$ is satisfied. The reason for the compatibility of the master equation $\Delta^h=0$ is readily revealed by examining the structure of $\Delta^h$ as given by \eqref{Z35}. Indeed, the discussion in the previous \high{subsection} has revealed that $\phi^a$ and $\phi^b$ are symmetries of the first heavenly equation and if the functions $a,b,c,d$ are suitable multiples of $y^3,y^4,y^1,y^2$ respectively then $\Delta^h|_{c^a=c^b=0}$ also constitutes a symmetry. In fact, in general, the condition \eqref{Z36} is exactly the condition which guarantees that $\Delta^h$ represents a symmetry of the first heavenly equation, \highb{that is, $\Delta^h$ satisfies the symmetry condition \eqref{Z24}, thereby justifying the notation $\Delta^h$.} Accordingly, the master equation $\Delta^h=0$ is nothing but a symmetry reduction of the first heavenly equation. 

In order to address the question as to whether the decomposition of the general heavenly equation into the dispersionless Hirota system corresponds to the assumption of a homothetic Killing vector, it is now required to determine whether the symmetry constraint \eqref{Z22} on the first heavenly equation constitutes a special case of the symmetry constraint $\Delta^h=0$. To this end, we first observe (as mentioned earlier, cf.\ Section 11.2) that insertion of $\Phi$ as given by \eqref{Z22} into the Lax pair \eqref{Z21} leads to two second-order constraints on the first heavenly equation. These constraints together with the first heavenly equation may then be formulated as a system of the type
\bela{Z38}
  \Theta_{y^1y^1} = F, \quad \Theta_{y^1y^3} = G,\quad \Theta_{y^3y^3} = H,\qquad F,G,H\in\mathcal{S}.
\ela
Here, the exact form of the functions $F,G$ and $H$ is not important. The key property of the above system is that these functions are contained in $\mathcal{S}$ which denotes the set of functions depending on $\Theta,\Theta_{y^1},\Theta_{y^3}$ and their derivatives of any order with respect to $y^2$ and $y^4$. The functions of this set may also depend explicitly on the independent variables $y^i$. By construction, the associated compatibility conditions are satisfied which, in turn, implies that, generically,  the solution of the triple \eqref{Z38} is determined by the Cauchy data
\bela{Z39}
  \Theta = f_0(y^2,y^4),\quad \Theta_{y^1} = f_1(y^2,y^4),\quad \Theta_{y^3} = f_3(y^2,y^4)\quad\mbox{at}\quad (y^1,y^3)=(y^1_0,y^3_0).
\ela
On the other hand, differentiation of $\Delta^h$ and use of \eqref{Z28} and \eqref{Z29} evidently yields \mbox{$\Delta^h_{y^2y^4}\in \mathcal{S}$}. Hence, $\Delta^h_{y^2y^4}=0$ evaluated at $(y^1,y^3)=(y^1_0,y^3_0)$ constitutes a differential constraint on the Cauchy data $f_k(y^2,y^4)$. Accordingly, generically, the constrained Pleba\'nski system \eqref{Z38}, which, as stated earlier, is equivalent to the dispersionless Hirota system, does not give rise to a homothetic Killing vector. Finally, it is well known \cite{Hall2004} that any metric satisfying Einstein's vacuum equations $R_{ik}=0$ which admits a proper conformal Killing vector also possesses a Killing vector. It is recalled that $V^k$ constitutes a proper conformal Killing vector if it satisfies Killing's equations \eqref{Z32} for a non-constant function $\chi$. This completes the proof of Theorem~\ref{killing}.

\section{Decomposition of the first heavenly equation}

In the previous section, it has been demonstrated that the dispersionless Hirota system \eqref{Z13}, \eqref{Z14} gives rise to self-dual Einstein spaces which, \high{generically}, do not admit any conformal Killing vectors. Since the dispersionless Hirota system has been obtained as a symmetry reduction of the general heavenly equation, it should be investigated whether the corresponding symmetry constraint may be formulated in an invariant manner so that it may be applied to any of the known heavenly equations. The results of this investigation will be presented elsewhere. Here, we briefly present the action of this symmetry constraint on the first heavenly equation since it has been exploited in the previous section.

\subsection{A first integral}

In Section 8.2, it has been shown that the potentials $\Xi$ and $\Theta$ obeying the general heavenly equation \eqref{N55} and the first heavenly equation \eqref{A19} respectively are linked by the second-order relations \eqref{A16}-\eqref{A18}. It turns out that an appropriate extension of these relations admits a first integral if the solutions of the general heavenly equation are restricted to the class of solutions satisfying the dispersionless Hirota system
\bela{A20}
 \begin{split}
  (\lambda-\mu_3)\Xi_{x^3}\Xi_{x^1x^2} - \lambda\Xi_{x^2}\Xi_{x^1x^3} + \mu_3\Xi_{x^1}\Xi_{x^2x^3} &= 0\\
  (\lambda-\mu_4)\Xi_{x^4}\Xi_{x^1x^2} - \lambda\Xi_{x^2}\Xi_{x^1x^4} + \mu_4\Xi_{x^1}\Xi_{x^2x^4} & = 0\\
 (\mu_3-\mu_4)\Xi_{x^1}\Xi_{x^3x^4} + (\lambda-\mu_3)\Xi_{x^3}\Xi_{x^1x^4} - (\lambda-\mu_4)\Xi_{x^4}\Xi_{x^1x^3} & = 0\\
 \lambda(\mu_3-\mu_4)\Xi_{x^2}\Xi_{x^3x^4} + \mu_4(\lambda-\mu_3)\Xi_{x^3}\Xi_{x^2x^4} - \mu_3(\lambda-\mu_4)\Xi_{x^4}\Xi_{x^2x^3} & = 0.
 \end{split}
\ela
Specifically, if we differentiate the ansatz
\bela{A21}
  \kappa\Xi = \Theta - y^1\Theta_{y^1} - y^3\Theta_{y^3}
\ela
with respect to $y^i$, where $\kappa$ is a constant, then we obtain the additional second derivatives 
\bela{A22}
 \begin{aligned}
 \Theta_{y^1y^1} &=-\frac{\kappa\Xi_{y^1}+y^3\Theta_{y^1y^3}}{y^1} ,\quad & \Theta_{y^1y^2} &=\frac{\Theta_{y^2} - \kappa\Xi_{y^2} - y^3\Theta_{y^2y^3}}{y^1}  \\
 \Theta_{y^3y^3} &=-\frac{\kappa\Xi_{y^3}+y^1\Theta_{y^1y^3}}{y^3},\quad & \Theta_{y^3y^4} &=\frac{\Theta_{y^4} - \kappa\Xi_{y^4} - y^1\highb{\Theta_{y^1y^4}}}{y^3}.
 \end{aligned}
\ela
It is evident that all terms of the right-hand sides of these relations except for $\Theta_{y^2}$ and $\Theta_{y^4}$ may be expressed entirely in terms of the potential $\Xi$ and the associated independent variables $x^i$. One may now directly verify (using computer algebra) that these relations are compatible \highb{with \eqref{A17}} modulo the dispersionless Hirota system \eqref{A20} provided that
\bela{A23}
 \kappa = -\frac{1}{2}\frac{\lambda}{(\lambda - \mu_3)(\lambda - \mu_4)}.
\ela
Hence, taking into account the invariance $\Xi\rightarrow \Xi + \mbox{const}$ of the general heavenly equation, the ansatz \eqref{A21} is equivalent to the system \eqref{A22} and may be interpreted as a first integral of the extended system \eqref{A17}, \eqref{A22}.

\subsection{Decomposition}

In order to motivate the first-order link \eqref{A21} between the potentials $\Xi$ and $\Theta$ in the case of the restricted class of solutions of the general heavenly equation governed by the dispersionless Hirota system, we note that since the latter is a result of the imposition of a symmetry constraint involving the eigenfunction and a scaling symmetry of the general heavenly equation, one expects to find a similar result if one formulates this symmetry constraint in terms of Pleba\'nski's first heavenly equation. In the previous section, we have demonstrated that the symmetry constraint
\bela{A24}
  \Phi = \Theta - y^1\Theta_{y^1} - y^3\Theta_{y^3}
\ela
on the first heavenly equation is admissible, where $\Phi$ obeys the Lax pair
\bela{A25}
 \begin{split}
  \Phi_{y^3} & = \highb{\bar{\lambda}}(\Theta_{y^1y^3}\Phi_{y^2} - \Theta_{y^2y^3}\Phi_{y^1})\\
  \Phi_{y^4} & = \highb{\bar{\lambda}}(\Theta_{y^1y^4}\Phi_{y^2} - \Theta_{y^2y^4}\Phi_{y^1})
 \end{split}
\ela
for the first heavenly equation which, in the current context, is of the form \eqref{A19}. Now, taking into account that $\Xi$ constitutes an eigenfunction if the general heavenly equation is specialised to the dispersionless Hirota system, the identification of the two eigenfunctions $\Phi$ and $\kappa\Xi$ leads to the ansatz \eqref{A21}.

Insertion of $\Phi$ as given by \eqref{A24} into the Lax pair \eqref{A25} leads to
\bela{A26}
 \begin{split}
 \Theta_{y^1y^1} & = \frac{(y^1\Theta_{y^1y^2} - \Theta_{y^2})\Theta_{y^1y^4} - (y^1\Theta_{y^1y^4} + y^3\Theta_{y^3y^4} - \Theta_{y^4})\highb{\bar{\lambda}}^{-1} - \sigma y^3}{y^1\Theta_{y^2y^4}} \\
 \Theta_{y^3y^3} & = \frac{(y^3\Theta_{y^3y^4} - \Theta_{y^4})\high{\Theta_{y^2y^3} + (y^1\Theta_{y^1y^2}} + y^3\Theta_{y^2y^3} - \Theta_{y^2})\sigma \highb{\bar{\lambda}} - \sigma y^1}{\highb{y^3}\Theta_{y^2y^4}} ,
 \end{split}
\ela
where we have exploited the Pleba\'nski equation \eqref{A19} formulated as
\bela{A27}
  \Theta_{y^1y^3} = \frac{\Theta_{y^2y^3}\Theta_{y^1y^4} + \sigma}{\Theta_{y^2y^4}}.
\ela
By construction, the system \eqref{A26}, \eqref{A27} is compatible. Finally, it may be verified that this system is satisfied if $\Theta$ is related to the potential $\Xi$ via the system \eqref{A17}, \eqref{A22} provided that
\bela{A28}
 \highb{\bar{\lambda}} = 2\frac{(\mu_3-\mu_4)^2(\lambda-\mu_3)}{\mu_3(\lambda-\mu_4)}.
\ela
Hence, the decomposition \eqref{A26}, \eqref{A27} of Pleba\'nski's first heavenly equation \eqref{A19} constitutes an incarnation of the decomposition of the general heavenly equation into the dispersionless Hirota system \eqref{A20}. \high{It is noted that \eqref{A28} may also be formulated as 
\bela{A28a}
  -(\sigma\highb{\bar{\lambda}})^{-1} = 2\frac{(\mu_3-\mu_4)^2(\lambda-\mu_4)}{\mu_4(\lambda-\mu_3)}
\ela
which highlights the symmetry of the pair \eqref{A26}.}

\section{Partial Legendre transformations. The TED equation}

In Section 6, it has been demonstrated that the general heavenly equation is invariant under a Legendre transformation. The analysis in the previous section naturally leads to the consideration of partial Legendre transformations applied to Pleba\'nski's first heavenly equation and, by extension, to the Husain-Park equation.

\subsection{A partial Legendre transform of the first heavenly equation}

The constraint \eqref{A24} suggests considering the partial Legendre transformation
\bela{A29}
 (\Theta; y^1,y^2,y^3,y^4)\quad \rightarrow\quad (\bar{\Theta};y_1,y^2,y_3,y^4),
\ela
where
\bela{A30}
 \bar{\Theta} = \Theta - y^1\Theta_{y^1} - y^3\Theta_{y^3},\quad y_i = \Theta_{y^i},
\ela
so that it is natural to examine the first heavenly equation \eqref{A19} in the form
\bela{A31}
  dy_3\wedge dy_4\wedge dy^3\wedge dy^4 = \sigma dy^1\wedge dy^2\wedge dy^3\wedge dy^4.
\ela
In terms of the new potential $\bar{\Theta}$, the definition of the variables $y_i$ formulated as
\bela{A32}
 d\Theta = y_idy^i
\ela 
implies that
\bela{A33}
 y^1 = -\bar{\Theta}_{y_1},\quad y_2 = \bar{\Theta}_{y^2},\quad y^3 = -\bar{\Theta}_{y_3},\quad y_4 = \bar{\Theta}_{y^4},
\ela
leading to the partial Legendre transform
\bela{A34}
 \bar{\Theta}_{y_1y_3}\bar{\Theta}_{y^2y^4} - \bar{\Theta}_{y^2y_3}\bar{\Theta}_{y_1y^4} = \sigma\big(\bar{\Theta}_{y_1y_1}\bar{\Theta}_{y_3y_3} - \bar{\Theta}_{y_1y_3}^2\big)
\ela
of the first heavenly equation. This form of the self-dual Einstein equations together with its Legendre-type connection with Pleba\'nski's first heavenly equation has been recorded in \cite{MalykhNutkuSheftel2003}. In the current context, the constraint \eqref{A24} shows that the counterpart of the decomposition of the general heavenly equation into the dispersionless Hirota system is the decomposition of the heavenly equation \eqref{A34} into a system analogous to the system \eqref{A26}, \eqref{A27} which is generated by matching the eigenfunction and \high{the potential $\bar{\Theta}$}, that is,
\bela{A35}
 \Phi = \bar{\Theta}
\ela
as in the case of the general heavenly equation.

\subsection{Connection with the 4+4-dimensional TED equation}

It turns out that the above observation is not a coincidence. As indicated in the preceding, the TED equation constitutes a 4+4-dimensional integrable generalisation of the general heavenly equation and, in fact, exists in $2n+2n$ dimensions~\cite{KonopelchenkoSchief2019}. The TED equation
\bela{A36}
 \begin{split}
  & \phantom{+}\,\,\, (\Theta_{y^1z^2} - \Theta_{y^2z^1})(\Theta_{y^3z^4} - \Theta_{y^4z^3})\\  
  & + (\Theta_{y^2z^3} - \Theta_{y^3z^2})(\Theta_{y^1z^4} - \Theta_{y^4z^1})\\
  & + (\Theta_{y^3z^1} - \Theta_{y^1z^3})(\Theta_{y^2z^4} - \Theta_{y^4z^2}) = 0
 \end{split}
\ela
is multi-dimensionally consistent \cite{KonopelchenkoSchief2019} and encodes many (if not all) known heavenly equations and also, for instance, the six-dimensional second heavenly equation (see, e.g.,  \cite{DoubrovFerapontov2010} and references therein). In particular, the travelling wave reduction
\bela{A37}
  \Theta_{z^i} = \lambda_i\Theta_{y^i}
\ela
leads to the general heavenly equation
\bela{A38}
 \begin{split}
  & \phantom{+}\,\,\, (\lambda_1-\lambda_2)(\lambda_3-\lambda_4)\Theta_{y^1y^2}\Theta_{y^3y^4}\\
  & +  (\lambda_2-\lambda_3)(\lambda_1-\lambda_4)\Theta_{y^2y^3}\Theta_{y^1y^4}\\  
  & +  (\lambda_3-\lambda_1)(\lambda_2-\lambda_4)\Theta_{y^3y^1}\Theta_{y^2y^4} = 0.
\end{split}
\ela
Moreover, it has been shown that any eigenfunction $\Phi$ obeying the associated Lax pair constitutes a symmetry of the TED equation so that matching this eigenfunction symmetry with the scaling symmetry of the TED equation encapsulated in the constraint $\Phi=\Theta$ leads to a higher-dimensional integrable extension of the dispersionless Hirota system consisting of four compatible 3+3-dimensional generalised dispersionless Hirota equations, namely
\bela{A39}
 \begin{split}
  & \phantom{+}\,\,\, (\Theta_{y^iz^k} - \Theta_{y^kz^i})(\Theta_{z^l} - \lambda\Theta_{y^l})\\  
  & + (\Theta_{y^kz^l} - \Theta_{y^lz^k})(\Theta_{z^i} - \lambda\Theta_{y^i})\\
  & + (\Theta_{y^lz^i} - \Theta_{y^iz^l})(\Theta_{z^k} - \lambda\Theta_{y^k}) = 0.
 \end{split}
\ela
Here, the indices $i,k,l\in\{1,2,3,4\}$ are distinct. Indeed, in the travelling wave reduction \eqref{A37}, the fully symmetric avatar of the dispersionless Hirota system is obtained. One may also directly verify that the TED equation \eqref{A36} is an algebraic consequence of the generalised dispersionless Hirota system \eqref{A39}.

Another travelling wave reduction of the TED equation generated by
\bela{A40}
 \begin{aligned}
 \Theta_{z^1} &= \lambda_1\Theta_{y^1},\quad &\Theta_{z^2} = &\lambda_1\Theta_{y^2} + \nu_2\Theta_{y^1}\\
 \Theta_{z^3} &= \lambda_3\Theta_{y^3},\quad &\Theta_{z^4} = &\lambda_3\Theta_{y^4} + \nu_4\Theta_{y^3}
 \end{aligned}
\ela
reads
\bela{A41}
  (\lambda_3-\lambda_1)^2(\Theta_{y^1y^3}\Theta_{y^2y^4} - \Theta_{y^2y^3}\Theta_{y^1y^4}) = \nu_2\nu_4(\Theta_{y^1y^1}\Theta_{y^3y^3} - \Theta_{y^1y^3}^2)
\ela
which is exactly of the form \eqref{A34} \highb{(with a slightly different labelling of the dependent and independent variables)}. Its decomposition into a system of compatible equations via the symmetry constraint $\Phi=\Theta$ is obtained by imposing the travelling wave constraints \eqref{A40} on the generalised dispersionless Hirota system \eqref{A39}. This explains why the heavenly equation \eqref{A34} admits the \high{(symmetry)} constraint \eqref{A35}. It is important to note that the travelling wave constraints \eqref{A40} constitute an extension of the travelling wave constraints \eqref{A37} subject to the choice $\lambda_1=\lambda_2$ and $\lambda_3=\lambda_4$. \high{The latter is precisely the specialisation employed in Section 8 which has led to Pleba\'nski's first heavenly equation}. In fact, this suggests that one should also consider the ``intermediate'' case given by
\bela{A42}
 \begin{aligned}
 \Theta_{z^1} &= \lambda_1\Theta_{y^1},\quad &\Theta_{z^2} = &\lambda_1\Theta_{y^2} + \nu_2\Theta_{y^1}\\
 \Theta_{z^3} &= \lambda_3\Theta_{y^3},\quad &\Theta_{z^4} = &\lambda_4\Theta_{y^4}
 \end{aligned}
\ela
and corresponding to the choice $\lambda_1=\lambda_2$ in the above-mentioned sense, leading to the reduction
\bela{A43}
 \bear{c}
  \Theta_{y^1y^3}\Theta_{y^2y^4} - \Theta_{y^2y^3}\Theta_{y^1y^4} = \tilde{\sigma}(\Theta_{y^1y^1}\Theta_{y^3y^4} - \Theta_{y^1y^3}\Theta_{y^1y^4})\as\dis \tilde{\sigma} = \frac{\nu_2(\lambda_4-\lambda_3)}{(\lambda_1 - \lambda_3)(\lambda_1 - \lambda_4)}
 \ear
\ela
of the TED equation. Indeed, application of the partial Legendre transformation
\bela{A44}
 (\Theta; y^1,y^2,y^3,y^4)\quad \rightarrow\quad (\tilde{\Theta};y_1,y^2,y^3,y^4),
\ela
with
\bela{A45}
 \tilde{\Theta} = \Theta - y^1\Theta_{y^1},\quad d\Theta = y_idy^i 
\ela
is readily shown to produce the Husain-Park equation
\bela{A46}
  \tilde{\sigma}\tilde{\Theta}_{y^3y^4} = \tilde{\Theta}_{y_1y^3}\tilde{\Theta}_{y^2y^4} - \tilde{\Theta}_{y^2y^3}\tilde{\Theta}_{y_1y^4}.
\ela
Hence, we conclude that the three sets of travelling wave constraints \eqref{A37}, \eqref{A40} and \eqref{A42} of the TED equation correspond to the three canonical choices of the parameters $\lambda_i$ in the current formalism with the associated reductions \eqref{A38}, \eqref{A41} and \eqref{A43} being linked to the general heavenly equation, Husain-Park equation and first heavenly equation by the respective (partial) Legendre transformation. 

In summary, a link between travelling wave reductions of the TED equation and the classification within the present formalism with respect to the number of pairs of coinciding spectral parameters has been established. This highlights, once again, the significance of the TED equation. In this connection, it is interesting to recall \cite{KonopelchenkoSchief2019} that, as mentioned at the beginning of this section, the Husain-Park and first heavenly equations may also be obtained directly from the TED equation by imposing appropriate constraints.

\subsubsection*{Acknowledgements}
W.K.S.\ wishes to express his gratitude to his colleague John Steele for sharing his expertise in the area of conformal Killing symmetries.


\begin{thebibliography}{99}

\bibitem{SKMHH2003} H.\ Stephani, D.\ Kramer, M.A.H.\ MacCallum, C.A.\ Hoenselaers and E.\ Herlt, {\sl Exact Solutions of Einstein's Field Equations}, 2nd edition, Cambridge University Press (2003). 

\bibitem{RogersSchief2002}
C.\ Rogers and W.K.\ Schief, {\sl B\"{a}cklund and Darboux Transformations. Geometry and Modern Applications in Soliton Theory}, Cambridge Texts in Applied Mathematics, Cambridge University Press (2002).

\bibitem{AblowitzSegur1981}
M.J.\ Ablowitz and H.\ Segur, {\sl Solitons and the Inverse Scattering Transform}, SIAM, Philadelphia (1981).

\bibitem{AblowitzClarkson1991}
M.J.\ Ablowitz and P.\ Clarkson, {\sl Solitons, Nonlinear Evolution Equations and Inverse Scattering}, Cambridge University Press (1991).

\bibitem{Fordy1990}
A.P.\ Fordy, ed, {\sl Soliton Theory: A Survey of Results}, Manchester University Press, Manchester and New York (1990).

\bibitem{Konopelchenko1990}
B.G.\ Konopelchenko, Soliton eigenfunction equations: the IST integrability and some properties, {\sl Rev.\ Math.\ Phys.\ }{\bf 2} (1990) 399--440.

\bibitem{OevelRogers1993}
W.\ Oevel and C.\ Rogers, Gauge transformations and reciprocal links in 2+1 dimensions, {\sl Rev.\ Math.\ Phys.\ }{\bf 5} (1993) 299--330. 

\bibitem{BogdanovKonopelchenko2013}
L.V.\ Bogdanov and B.G.\ Konopelchenko, Grassmannians Gr($N$-1,$N$+1), closed differential $N$-1-forms and $N$-dimensional integrable systems, {\sl J.\ Phys.\ A: Math.\ Theor.\ }{\bf 46} (2013) 085201 (17pp).

\bibitem{ManakovSantini2006}
S.V.\ Manakov and P.M.\ Santini, Inverse scattering problem for vector fields and the Cauchy problem for the heavenly equation, {\sl Phys.\ Lett.\ A} {\bf 359} (2006) 613--619.

\bibitem{ZakharovShabat1979}
V.E.\ Zakharov and A.B.\ Shabat, Integration of nonlinear equations of mathematical physics by the method of inverse scattering, II. {\sl Func.\ Anal.\ Appl.\ }{\bf 13} (1979) 166--174.

\bibitem{Boyer1983}
C.P.\ Boyer, The geometry of complex self-dual Einstein spaces, in K.B.\ Wolf, ed, {\sl Nonlinear Phenomena}, Lecture Notes in Physics {\bf 189} (1983) 25--546.

\bibitem{Schief1996}
W.K.\ Schief, Self-dual Einstein spaces via a permutability theorem for the Tzitzeica equation, {\sl Phys.\ Lett.\ A} {\bf 223} (1996) 55--62. 

\bibitem{Schief1999}
W.K.\ Schief , Self-dual Einstein spaces and a discrete Tzitzeica equation. A permutability theorem link, in P.\ Clarkson and F.\ Nijhoff, eds, {\sl Symmetries and Integrability of Difference Equations}, London Mathematical Society, Lecture Note Series {\bf 255}, Cambridge University Press (1999) 137--148.

\bibitem{Takasaki1989}
K.\ Takasaki, Aspects of integrability in self-dual Einstein metrics and related equations, {\sl Publ.\ RIMS, Kyoto Univ.\ }{\bf 22} (1986) 949--990.

\bibitem{Plebanski1975}
J.F.\ Pleba\'nski, Some solutions of complex Einstein equations, {\sl J.\ Math.\ Phys.\ }{\bf 16} (1975) 2395--2402.

\bibitem{JakimowiczTafel2006} M.\ Jakimowicz and J.\ Tafel, Self-dual metrics in Husain’s approach, {\sl Class.\ Quantum Grav.\ }{\bf 23} (2006) 4907--4914.

\bibitem{KonopelchenkoSchief2019}
B.G.\ Konopelchenko and W.K.\ Schief, On an integrable multi-dimensionally consistent $2n+2n$-dimensional heavenly-type equation, {\sl Proc.\ R.\ Soc.\ London A} {\bf 475} (2019) 20190091 (21pp).

\bibitem{Krynski2018}
W.\ Kry\'nski, On deformations of the dispersionless Hirota equation, {\sl J.\ Geom.\ Phys.\ }{\bf 127} (2018) 46--54.

\bibitem{Sergyeyev2017}
A.\ Sergyeyev, A simple construction of recursion operators for multidimensional dispersionless integrable systems, {\sl J.\ Math.\ Anal.\ Appl.\ }{\bf 454} (2017) 468--480.

\bibitem{DunajskiKrynski2014} 
M.\ Dunajski and W.\ Kry\'nski, Einstein–Weyl geometry, dispersionless Hirota equation and Veronese webs, {\sl Math.\ Proc.\ Camb.\ Phil.\  Soc.\ }{\bf 157} (2014) 139--150.

\bibitem{MasonNewman1989}
L.J.\ Mason and E.T.\ Newman, A connection between the Einstein and Yang-Mills equations, {\sl Commun.\ Math.\ Phys.\ }{\bf 121} (1989) 659--668.

\bibitem{SibataMorinaga1936}
 T.\ Sibata and K.\ Morinaga, Complete and simpler treatment of wave geometry, {\sl J.\ Sci.\ Hiroshima Univ.\ Ser. A} {\bf 6} (1936) 173--189.

\bibitem{DoubrovFerapontov2010}
B.\ Doubrov and E.V.\ Ferapontov, On the integrability of symplectic Monge-Amp\`ere equations.
{\sl J.\ Geom.\ Phys.\ }{\bf 60} (2010) 1604--1616.

\bibitem{HodgePedoe1994}
W.V.D.\ Hodge and D.\ Pedoe, {\sl Methods of Algebraic Geometry. Vol.\ 1}, Cambridge Mathematical Library, Cambridge University Press (1994).

\bibitem{BobenkoSchief1999}
A.I.\ Bobenko and W.K.\ Schief, Discrete indefinite affine spheres, in A.\ Bobenko and R.\ Seiler, eds, {\sl Discrete Integrable Geometry and Physics}, Oxford University Press (1999) 113--138.

\bibitem{PlebanskiPrzanowski1996}
J.F.\ Pleba\'nski and M.\ Przanowski, The Lagrangian of a self-dual gravitational field as a limit of the SDYM Lagrangian, {\sl Phys.\ Lett.\ A} {\bf 212} (1996) 22-28.

\bibitem{Straightening} S.\ Sternberg,  {\sl Lectures on Differential Geometry}, AMS Chelsea Publishing, Providence, Rhode Island  (1999).

\bibitem{JonesTod1985}
P.E.\ Jones and K.P.\ Tod, Minitwistor spaces and Einstein-Weyl spaces, {\sl Class.\ Quantum Grav.\ }{\bf 2} (1985) 565--577.

\bibitem{Bogdanov2015}
L.V.\ Bogdanov, Doubrov-Ferapontov general heavenly equation and the hyper-K\"ahler hierarchy, {\sl J.\ Phys.\ A: Math.\ Theor.\ }{\bf 48} (2015) 235202 (15pp).


\bibitem{MalykhNutkuSheftel2003}
A.A.\ Malykh, Y.\ Nutku and M.B.\ Sheftel, Partner symmetries of the complex Monge-Amp\'ere
equation yield hyper-K\"ahler metrics without continuous symmetries, {\sl J.\ Phys.\ A: Math.\ Gen.\ }{\bf 36} (2003) 10023--10037.

\bibitem{BoyerFinley1982}
C.P.\ Boyer and J.D.\ Finley, III, Killing vectors in self-dual Euclidean Einstein spaces, {\sl J.\ Math.\ Phys.\ }{\bf 23} (1982) 1126--1130.

\bibitem{Hall2004}
G.\ Hall, {\sl Symmetries and Curvature Structure in General Relativity}, World Scientific, Singapore (2004).

\end{thebibliography}
\end{document}